\documentclass[conference]{IEEEtran}
\IEEEoverridecommandlockouts
\usepackage{cite}
\usepackage{amsthm}
\usepackage{amsmath,amssymb,amsfonts}
\usepackage[linesnumbered, ruled, longend]{algorithm2e}
\usepackage{algorithmic}
\usepackage{graphicx}
\usepackage{textcomp}
\usepackage{xcolor}
\usepackage{framed}
\usepackage{diagbox}
\usepackage{url}
\usepackage{multirow}
\usepackage{bm}
\usepackage[cmintegrals]{newtxmath}
\usepackage{caption, subcaption}
\usepackage{tabularx}
\usepackage{booktabs}
\newtheorem{theorem}{\bf Theorem}

\newtheorem{lemma}{Lemma}

\newtheorem{obs}{Observation}
\newtheorem{definition}{\bf Definition}

\newtheorem{assumption}{Assumption}

\newtheorem{counter-intuition}{Counter-intuition}

\newtheorem{insight}{Insight}
\def\Nsetp{\mathbb{N}_+} 

\def\Ex{\mathbb{E}}

\def\Tl{Low-type}
\def\Th{High-type}
\def\Lcost{c_L}
\def\Hcost{c_H}

\def\Lqual{r_L}
\def\Hqual{r_H}
\def\None{$N$}
\def\Ltype{$L$}
\def\Htype{$H$}
\def\Both{$B$}
\def\ageloss{\delta^2}
\def\cel{\eta_L}
\def\ceh{\eta_H}
\def\ceb{\eta_B}
\def\dis{\rho}

\ifodd 0
\newcommand{\rev}[1]{{\color{blue}#1}} 
\newcommand{\com}[1]{\textbf{\color{red} (COMMENT: #1) }} 
\newcommand{\comg}[1]{\textbf{\color{green} (COMMENT: #1)}}
\newcommand{\response}[1]{\textbf{\color{green} (RESPONSE: #1)}} 
\else
\newcommand{\rev}[1]{#1}

\newcommand{\com}[1]{}
\newcommand{\comg}[1]{}
\newcommand{\response}[1]{}
\fi

\def\BibTeX{{\rm B\kern-.05em{\sc i\kern-.025em b}\kern-.08em
    T\kern-.1667em\lower.7ex\hbox{E}\kern-.125emX}}
\begin{document}


\title{Recruiting Heterogeneous Crowdsource Vehicles for Updating a High-definition Map
}

\author{\IEEEauthorblockN{Wentao Ye, Yuan Luo, Bo Liu, Jianwei Huang
    \thanks{Wentao Ye is with Shenzhen Institute of Artificial Intelligence and Robotics for Society (AIRS), School of Science and Engineering, The Chinese University of Hong Kong, Shenzhen, Guangdong, China (e-mail: wentaoye@link.cuhk.edu.cn). }
    \thanks{ Yuan Luo is with School of Science and Engineering, The Chinese University of Hong Kong, Shenzhen, Guangdong, China (e-mail: luoyuan@cuhk.edu.cn).}
    \thanks{ Bo Liu is with Shenzhen Institute of Artificial Intelligence and Robotics for Society (AIRS), Guangdong, China (e-mail: liubo@cuhk.edu.cn).}
    \thanks{Jianwei Huang is with School of Science and Engineering, Shenzhen Institute of Artificial Intelligence and Robotics for Society, The Chinese University of Hong Kong, Shenzhen, Shenzhen 518172, China (corresponding author, e-mail: jianweihuang@cuhk.edu.cn).}
   }
   }
\maketitle

\begin{abstract}
\rev{The high-definition map is a cornerstone of autonomous driving. Unlike constructing a costly fleet of mapping vehicles, the crowdsourcing paradigm is a cost-effective way to keep an HD map up to date.} \rev{Achieving practical success for crowdsourcing-based HD maps is contingent on addressing two critical issues: freshness and recruitment costs. Given that crowdsource vehicles are often heterogeneous in terms of operational costs and sensing capabilities, it is practical to recruit heterogeneous crowdsource vehicles to achieve the tradeoff between freshness and recruitment costs. However, existing works neglect this aspect. To solve it, we formulate this problem as a Markov decision process.} We demonstrate that the optimal policy is threshold-type age-dependent. \rev{Additionally, our findings reveal some counter-intuitive insights. In some cases, the company should initiate vehicle recruitment earlier when vehicles arrive more frequently, or have higher operational costs or sensing capabilities.} \rev{Besides, we propose an efficient algorithm, called the bound-based relative value iteration (BRVI) algorithm, to overcome the technical challenge that finding an optimal policy is time-consuming.} Numerical simulations show that (i) \rev{the optimal policy reduces the average cost by $19.04\%$ compared to the state-of-the-art mechanism}, and (ii) \rev{the proposed algorithm can reduce the convergence time by $13.66\%$ on average compared to the existing algorithm.}
\end{abstract}



		

\section{Introduction}
\rev{The high-definition (HD) map is a foundation for autonomous driving since it provides autonomous vehicles with all critical static and dynamic information about their surrounding environment\cite{seif2016autonomous,bao2022high}.} Dynamic information in HD maps such as the \rev{congestion resolution}, dynamic incidence (\emph{e.g.,} agglomerate fog), and the construction area\cite{massow2016deriving}, is crucial. According to their real-time and unpredictable nature, dynamic information in HD maps must be highly fresh, requiring updates from the HD map company (``company" hereafter) within seconds or minutes\cite{b2}. Traditional methods for updating HD maps involve dedicated mapping vehicles equipped with high-precision sensors, such as LiDARs and high-accuracy cameras\rev{\cite{Ria2015build}}. However, maintaining such a fleet is expensive, limiting both its scale and the frequency of updates.

To reduce the cost while maintaining the information freshness, many \rev{HD map companies have adopted crowdsourcing, which is often viewed as a cheap source of high-freshness data\cite{Mobileye2017Crowd, Here2021Mer}.} However, there are still some issues to be addressed. \rev{On one hand, to ensure driver safety and promote efficiency in the transportation system, it is desirable to have an HD map that is as fresh as possible\cite{hao2022freshness}. To achieve this, companies\footnote{For convenience, we refer to the company as 'she' and vehicles as 'he'.} need to recruit a significant number of crowdsource vehicles to collect fresh data and pay for their operational cost. These payments comprise the company's recruitment costs. On the other hand, although crowdsourcing is often viewed as a cheap source for collecting data, keeping an up-to-date HD map can be prohibitively expensive, especially for high-frequency updates and wide coverage HD maps.\cite{cao2019online}. Thus, achieving a tradeoff between freshness and recruitment costs is crucial for the practical success of crowdsourcing-based HD maps.}

\rev{Achieving the tradeoff between freshness and recruitment costs requires considering the heterogeneity of crowdsource vehicles in terms of their operational costs and sensing capabilities. This is practical since vehicles from different manufacturers often have varying operational costs and sensing capabilities due to their installed sensors. Furthermore, according to \cite{huang2019crowdsourcing}, the heterogeneity of crowdsource workers' costs and capability usually has a great impact on a crowdsourcing platform's tradeoff.}

\rev{While few studies focus on freshness and recruitment costs, some\cite{chen2022love,cao2020trajectory} only consider one aspect of the tradeoff and ignore the other, potentially leading to cost-ineffective methods for creating and updating HD maps. Other studies\cite{zhang2019mobile, shi2023federated} that do consider both aspects of the tradeoff overlook the heterogeneity of crowdsource vehicles, which prevents practical optimization of the tradeoff. Our work focuses on achieving the tradeoff between freshness and recruitment costs while accounting for heterogeneous crowdsource vehicles.}

\rev{To achieve this tradeoff, we propose minimizing the company's time-average cost (``average cost" hereafter).} \rev{We define this average cost by the sum of two components: (1) the company's profit loss caused by an outdated HD map and (2) recruitment costs. Since an outdated HD map will incur some profit loss, we can use a loss function of the age of information (AOI) called AoI loss to measure this profit loss\cite{wang2022dynamic}.} In sum, we characterize the tradeoff by the company's average cost, which comprises AoI loss and recruitment costs. 

\rev{Besides, in the context of crowdsourcing for HD maps, the technical challenge is an efficient derivation of the optimal recruitment policy. Because Short-term traffic flow prediction is generally more accurate than long-term\cite{yu2016data,wang2019traffic,majumdar2021congestion}, with "short-term" referring to a few minutes, such as 2 minutes\cite{rusyaidi2020review}. Thus, the company must derive and apply the optimal policy quickly. However, existing related works often neglect this aspect.}

Against this background, we design an MDP model to minimize the average cost with heterogeneous crowdsource vehicles. \rev{Our model accounts for the varying sensing capabilities and operational costs of crowdsource vehicles and assumes that they arrive at a single place of interest (PoI) randomly.} The company \rev{needs to} decide whether to recruit arriving vehicles to minimize her average cost. To determine the optimal recruiting policy within an efficient computational time, we formulate the decision process as a Markov decision process (MDP) and design a time-efficient algorithm called the BRVI algorithm.

This paper advances the state of the art in the following ways: 
\begin{itemize}
    \item \textit{A Novel Problem Formulation:} To the best of our knowledge, it is the first study \rev{to examine the tradeoff between AoI loss and recruitment costs of HD maps involving crowdsource vehicles.} To be practical, we assume crowdsource vehicles arrive randomly and have heterogeneous operational costs and sensing capabilities. Achieving such a tradeoff can help the company reduce the average cost of maintaining an HD map.
    \item \textit{Optimal Policy Analysis:} We prove that the optimal policy \rev{is threshold-type age-dependent. Additionally, our findings reveal some counter-intuitive insights. In some cases, the company should initiate vehicle recruitment earlier when vehicles arrive more frequently, or have higher operational costs or sensing capabilities.}
    \item \textit{An Efficient Algorithm:} We propose an algorithm called the BRIV algorithm that derives the optimal policy efficiently. \rev{Compared to the existing algorithm based on the relative value iteration (RVI), it leverages the upper bounds of the corresponding thresholds to reduce the feasible set of optimization operations. This ultimately saves convergence time.}
    \item \textit{Numerical Results:} Numerical results show that the optimal policy \rev{reduces the company's average cost by up to $19.04\%$ compared with the state-of-the-art mechanism}\cite{wang2022dynamic}, respectively. Moreover, the proposed BRVI algorithm reduces the convergence time by $13.66\%$ on average compared to the existing algorithm.
\end{itemize}

\section{Related Works}
This section covers two related areas: \rev{mobile crowdsensing related to AoI} and crowdsourcing-based HD map updates.

rec{In the field of AoI-related mobile crowdsensing, there recently has been some work jointly considering the freshness and recruitment costs issues\cite{krishnan2014robust,yu2020edge,wang2022dynamic}. For instance, Wang \textit{et al}. proposed dynamic pricing to offer age-dependent monetary rewards, thus encouraging workers to contribute\cite{wang2022dynamic}. However, they assume crowdsource workers have homogeneous capabilities, which is not practical for crowdsourcing-based HD map settings. Our work involves vehicles' heterogeneity in terms of sensing capabilities when considering the tradeoff between AoI and recruitment cost.

While the tradeoff between freshness and recruitment cost is well-studied in AoI-related mobile crowdsensing, there are only a few papers considering these factors separately in crowdsourcing-based HD map updates. For example, Shi \textit{et al.} formulated an overlapping coalition formation game to recruit crowdsource vehicles to contribute data and improve the HD map quality, which is defined by freshness and spatial completeness\cite{shi2023federated}. Li \textit{et al}. developed a periodic crowdsource task distribution framework to achieve high time coverage and efficiency with lower recruitment cost\cite{li2021brief}. Our work is the first to consider the tradeoff between AoI loss and recruitment expenses in crowdsourcing-based HD map updates, which has not been addressed in the previous papers.}

\section{System Model}
In this section, we present the system model to minimize the company's average cost. \rev{First, we describe the crowdsourcing paradigm in Section~\ref{subsec: Overview of crowdsourcing in the HD map}, then introduce the company's cost function in Section~\ref{subsec: average cost}.} Finally, we formulate the average cost minimization problem as an MDP in Section~\ref{subsec: formulation}.

\subsection{Overview of crowdsourcing in the HD map} \label{subsec: Overview of crowdsourcing in the HD map}

We consider a setting in that an HD map company collects dynamic information at a PoI from crowdsource vehicles \rev{over a period of time $t\in\Nsetp$}. For ease of presentation, we assume there are only two types of crowdsource vehicles: \rev{\Tl}~and \rev{\Th}\footnote{For a more general case where vehicles' operational costs and sensing capabilities follow arbitrary distributions, our results remain valid.}. \rev{\Tl~vehicles have worse sensors than \Th~vehicles, meanwhile the corresponding operational costs and sensing capabilities satisfy $\Lcost<\Hcost$ and $\Lqual<\Hqual$, respectively.} Here, \rev{we use $\Lqual$~and $\Hqual$~to denote the sensing capability of \Tl~and \Th~vehicle, which refers to the probability of acquiring qualified sensing data.} \rev{Unqualified sensing data refers to data that the company cannot confidently use to update HD maps.} Note that the company can detect the quality of sensing data before using it\footnote{Take images acquired by cameras as an example, there are many preceding methods to detect the quality of the image, such as NR-IQA\cite{kang2014convolutional} and Nanny ML\cite{humphrey2022machine}.}. If sensing data are not qualified, the company will not use them to update the HD map.

\rev{As is common in the literature (\emph{e.g.,}\cite{zhang2018monte,cai2009adaptive}),} we use the Bernoulli process to model crowdsource vehicles' arrival for analytical convenience\footnote{A more general model is the Markov arrival process, which largely expands the state space and complicates the problem. We will consider the general scenario in our future work.}. Hence, we make Assumption \ref{ass: arrival process} as follows.
\begin{assumption}\label{ass: arrival process}
The arrivals of \Tl~and \Th~vehicles are independent and follow the Bernoulli processes with the arrival probabilities of $p_L$ and $p_H$, respectively.
\end{assumption}

In this subsection, we briefly describe the crowdsourcing workflow in the HD map.
\begin{figure}[htbp]
    \centering
    \includegraphics[width=6cm]{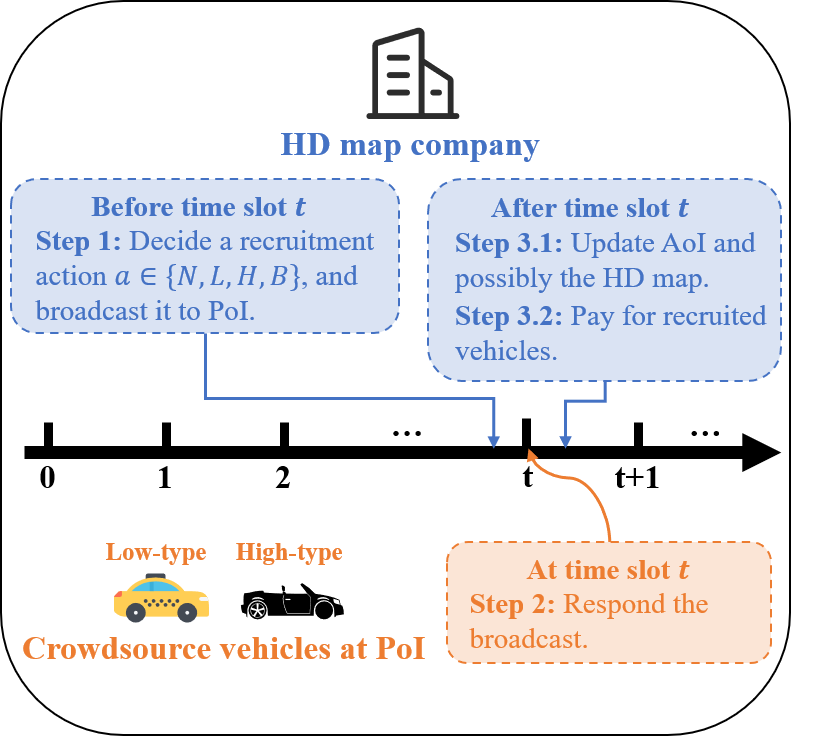}
    \caption{\rev{The crowdsourcing workflow for the HD map.}}
    \label{fig: crowdsourcing_workflow}
\end{figure}

Fig.~\ref{fig: crowdsourcing_workflow} illustrates the workflow of the HD map crowdsourcing in each time slot $t$. The details are as follows: 
\begin{itemize}        
    \item In stage 1, the company announces a recruitment action $a\in\{N, L, H, B\}$ \rev{before the time slot $t$}, where $N$ means no recruitment, $L$ means only recruiting a \Tl~vehicle, $H$ means only recruiting a \Th~vehicle, and $B$ means recruiting both.
    \item In stage 2, \rev{crowdsource vehicles arrive at PoI and respond to the broadcast in the time slot $t$}. An example scenario involves a broadcast message stating ``the recruitment action is \Ltype." In this case, the \Tl~vehicle present at the PoI will detect the PoI and upload the corresponding sensing data to the company since she will pay to cover the operational cost of the vehicle.   
    \item In stage 3, the company finishes the following things \rev{after the time slot $t$}:
    \begin{itemize}
        \item \rev{In stage 3.1, she updates the HD map and the corresponding AoI once receiving the qualified sensing data from recruited vehicles.}
        \item \rev{In stage 3.2, she pays for the recruited vehicles once receiving sensing data from recruited vehicles}\footnote{Because recruited vehicles have already incurred operational costs to acquire and send sensing data. If the company refuses to reward recruited vehicles that provide unqualified sensing data, she would discourage broad participation.}.
    \end{itemize}
         
\end{itemize}

\subsection{Company's cost function} \label{subsec: average cost}
In this subsection, we first describe the AoI loss. Next, we introduce \rev{the company's cost function that depends on the current AoI and her action.}

\subsubsection{AoI loss} Similar as the literature\cite{wang2018microeconomic}, we use AoI loss to measure the freshness. The less fresh the HD map, the higher value of AoI loss. Let $\delta(t)$ denote the current AoI at time slot $t$. We choose the square function $\delta^2(t)$, which reflects AoI loss that the company's cost convexly increases in the age of its provided information\cite{wang2018microeconomic}.

Based on the definition of AoI\cite{hsu2019scheduling, pan2021minimizing}, the AoI increases linearly as time goes on and resets to 1 when the company updates the HD map. As a result, the dynamics of AoI $\delta(t)$ is as follows
\rev{\begin{equation}\label{Eq: dynamics of AoI}
    \delta(t+1) = \left\{\begin{array}{ll}
        1,  & \text{the HD map is updated,} \\
        \delta(t)+1, & \text{otherwise}.
    \end{array}
    \right.
\end{equation}}

The company updates the HD map when it successfully recruits at least one vehicle that provides qualified sensing data in a time slot. \rev{More specifically, the update happens in a time slot if one of the following is true:}
\begin{itemize}
    \item \rev{The company chooses the recruitment action \Ltype; meanwhile, a \Tl~vehicle arrives and reports the qualified sensing data; }
    \item \rev{The company chooses the recruitment action \Htype; meanwhile, a \Th~vehicle arrives and reports the qualified sensing data; }
    \item \rev{The company chooses the recruitment action \Both; meanwhile, a \Tl~or \Th~vehicle arrives and reports the qualified sensing data.}
\end{itemize}

\subsubsection{Company's cost function} The company's cost function is a weighted sum of the AoI loss and the recruitment costs. Due to the uncertainty introduced by random arrivals of crowdsource vehicles and the quality of sensing data, the company should minimize the expected cost.

First, \rev{we define the success probability $Q_a$ as the probability that the company will successfully update the HD map when choosing the action $a\in\mathcal{A}$.}
\begin{equation} \label{eq: success probability}
    Q_a = \left\{\begin{array}{ll}
        0,  &  a=N,\\
        r_ap_a,  & a=L\text{ or } H,\\
        r_Lp_L + r_Hp_H - r_Lr_Hp_Lp_H, & a=B.
    \end{array}\right. 
\end{equation}
where \rev{the success probability $Q_B$ of taking action \Both~is the probability that the union of two non-mutually exclusive events, \emph{i.e.,} a \Tl~or \Th~vehicle reports qualified sensing data.} 

\rev{Then, we introduce the expected AoI loss and expected recruitment costs given an action $a\in\mathcal{A}$.} Suppose that the current AoI is $\delta(t)$, then we define
\begin{itemize}
    \item \textit{Expected AoI loss} as $(1-Q_a)\delta(t)^2$, where $Q_a$ is the success probability defined in Eq.~\eqref{eq: success probability}.
    \item \textit{Expected recruitment costs} as $p_ac_a$, where 
    rev{\begin{equation}
        p_ac_a = \left\{\begin{array}{ll}
            0,  &  a=N,\\
            p_ac_a,  & a=L\text{ or } H,\\
            p_Lc_L+p_Hc_H, & a=B.
        \end{array}\right. 
    \end{equation}
    Notice that for action $a=B$, the expected recruitment cost is $p_Bc_B=p_Lc_L+p_Hc_H$ as the arrivals of \Tl~and \Th~vehicles are independent.}
\end{itemize}

Finally, we can characterize the company's cost function as follows. Suppose the AoI is $\delta=\delta(t)$ and the chosen action is $a\in\mathcal{A}$.
\begin{equation}\label{eq: cost function}
    \begin{split}
        u(\delta,a)
        =\left\{
        \begin{array}{ll}
            \beta\delta^2, & a=N,\\
            (1-\beta) p_Lc_L + \beta(1-Q_L)\delta^2,  & a=L,  \\
            (1-\beta) p_Hc_H + \beta(1-Q_H)\delta^2, & a=H,  \\
            (1-\beta)p_Bc_B + \beta(1-Q_B)\delta^2, & a=B,  \\    
        \end{array}\right.
    \end{split}    
\end{equation}
where \rev{$\beta$} is the weighted factor to balance the AoI loss and recruitment costs. The larger $\beta$ is, the more emphasis we place on the freshness.


\subsection{Problem Formulation} \label{subsec: formulation}
In this subsection, we formulate the average cost minimization problem as an average-cost MDP \rev{$\Lambda$}.

\begin{itemize}
    \item \textbf{States}: \rev{$S(t)=\delta(t)\in\mathcal{S}$, \emph{i.e.,} the AoI in time slot $t$. Let $\mathcal{S}=\Nsetp$ be the state space, which is countably infinite.}
    \item \textbf{Actions}: \rev{$a\in\mathcal{A}$, where $\mathcal{A}=\{N,L,H,B\}$ is the action space.}
    \item \textbf{State transition probability}: \rev{$P_{ss'}(a)$, \emph{i.e.,} the state transition probability from state $s$ to $s'$ under the action $a$.} Suppose the state $s$ is $\delta(t)$, the non-zero $P_{ss'}(a)$ is as follows\\
    \vspace{-0.3cm}
    \begin{equation} 
    P_{ss'}(a)=\left\{
    \begin{array}{ll}
        1,&s'=\delta(t)+1,a=N,\\            
        Q_a, &s'=1,a\not=N,\\
        1-Q_a,&s'=\delta(t)+1,a\not=N.\\
    \end{array}\right.
    \end{equation}
    \item \textbf{Cost}: \rev{$u(\delta=\delta(t), a)$ as in Eq.~\eqref{eq: cost function}, \emph{i.e.,}  the company's cost if she takes the action $a$ at slot $t$ under state $S(t)=\delta(t)$, }
    \item \textbf{Policy}: \rev{$\phi=\{\psi_1(\cdot),\psi_2(\cdot),...,\psi_{\infty}(\cdot)\}$, where $\psi_t(\cdot)$ is a function mapping from $\mathcal{S}$ to $\mathcal{A}$.}
\end{itemize}

\rev{The company's average cost under a policy $\phi$ is the value function $J(\phi)$ of the MDP $\Lambda$ under this policy as follows.}

\begin{equation}\label{Average cost}
    J(\phi)=\lim_{T\to\infty} \frac{1}{T} \mathbb{E}^\phi\left[\sum_{t=1}^T u\left(\delta=\delta(t),a=\psi_t(S(t))\right)\right].
\end{equation}

Our objective is to minimize the company's average cost $J(\phi)$ by finding the optimal policy $\phi^*$. 

\rev{According to \cite{puterman2014markov}, identifying the optimal deterministic stationary policy is sufficient for finding the optimal policy of MDP $\Lambda$.} 

\rev{\begin{definition}
    A \emph{deterministic stationary} policy is a policy that selects actions with certainty based only on the current state of the system.
\end{definition}

Note that we can regard the optimal policy $\phi^*$ as a function mapping from state space $\mathcal{S}$ to the action space $\mathcal{A}$.}


\section{Main Results}\label{sec: main results}
In the section, we characterize the optimal policy and provide an \rev{efficient} algorithm to find it. Subection~\ref{subsec: Characterization of the optimality}, shows that the optimal policy \rev{is threshold-type age-dependent}, and also provides some counter-intuitive insights. In Section~\ref{subsec: BRVIA}, we \rev{propose the BRVI algorithm to efficiently find the optimal policy by characterizing the upper bounds of optimal thresholds.}

\rev{We first define the \textit{cost-effectiveness}, which is useful to determine the structure of the optimal policy and to present the upper bounds of the thresholds.}
\begin{definition}
    \rev{The cost-effectiveness of} action $a\in\mathcal{A}/\{N\}$ by $\eta_a$ is the ratio between the expected recruitment costs and the success probability of this action.
    \begin{equation}
        \eta_a = \frac{p_a c_a}{Q_a}.
    \end{equation}
\end{definition}

For example, the cost-effectiveness of action $L$ is $\frac{p_Lc_L}{Q_L}=\frac{p_Lc_L}{p_Lr_L}=\frac{c_L}{r_L}$. \rev{Note that the smaller the cost-effectiveness, the more cost-effective the corresponding action becomes.}

\subsection{Characterization of the optimality}\label{subsec: Characterization of the optimality}

In this subsection, we characterize four \rev{possible} structures of the optimal policy in Theorem~\ref{theorem: 1} and illustrate some counter-intuitive insights.

\rev{
\begin{theorem}\label{theorem: 1}
    There exists a threshold-type age-dependent stationary deterministic optimal policy $\phi^*$ of the MDP $\Lambda$. Depending on the parameter setting, such an optimal policy falls into one of the following four structures:
    \begin{enumerate}
        \item \textbf{LH structure} (in Fig.~\ref{fig: LH structure}): if $\frac{Q_L}{Q_H} \leq 1$ and $\frac{1-Q_H}{1-Q_L}<\frac{\eta_L}{\eta_H}<1$, the optimal policy includes actions in the order of \None, \Ltype, \Htype, \Both, with three thresholds satisfying $1 \leq \delta_L \leq \delta_H \leq \delta_{B}\in\Nsetp$.
        \item \textbf{HL structure} (in Fig.~\ref{fig: HL structure}): if $\frac{Q_L}{Q_H} > 1$ and $1<\frac{\eta_L}{\eta_H}<\frac{1-Q_H}{1-Q_L}$, the optimal policy includes actions in the order of \None, \Htype, \Ltype, \Both, with three thresholds satisfying $1\leq \delta_H \leq \delta_L \leq \delta_{B}\in \Nsetp$.        
        \item \textbf{None-L structure} (in Fig.~\ref{fig: None-L structure}): if $\frac{\eta_L}{\eta_H} \geq \max\{1,\frac{1-Q_H}{1-Q_L}\}$, the optimal policy includes actions in the order of \None, \Htype, \Both, with two thresholds satisfying $1 \leq \delta_H \leq \delta_{B} \in \Nsetp$.        
        \item \textbf{None-H structure} (in Fig.~\ref{fig: None-H structure}): if $\frac{\eta_L}{\eta_H}<\min\{1,\frac{1-Q_H}{1-Q_L}\}$, the optimal policyincludes actions in the order of \None, \Ltype, \Both, with two thresholds satisfying $1 \leq \delta_L \leq \delta_{B} \in \Nsetp$.
    \end{enumerate}
    \begin{figure}[h]
        \centering
        \subcaptionbox{\rev{LH structure}\label{fig: LH structure}}{\includegraphics[width = .5\linewidth]{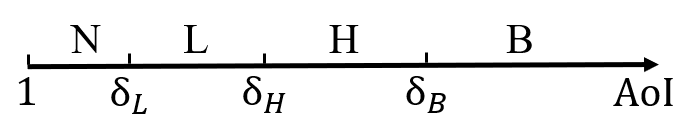}}\hfill
        \subcaptionbox{HL structure\label{fig: HL structure}}{\includegraphics[width = .5\linewidth]{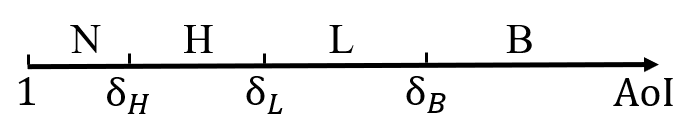}}
        \quad
        \subcaptionbox{None-L structure\label{fig: None-L structure}}{\includegraphics[width = .5\linewidth]{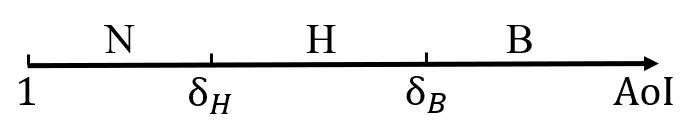}}\hfill
        \subcaptionbox{HL structure\label{fig: None-H structure}}{\includegraphics[width = .5\linewidth]{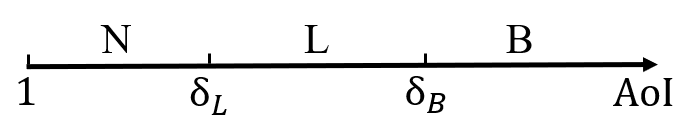}}
        \caption{\rev{Four possible structures of the optimal policy.}}
        \label{fig: possible structures}
    \end{figure}
    \vspace{-0.4cm}
\end{theorem}}

Due to the limited space, all the proofs are provided in the online appendix\cite{appendix}.

\begin{figure}[h]
    \centering
    \vspace{-0.4cm}
    \includegraphics[width = .6\linewidth]{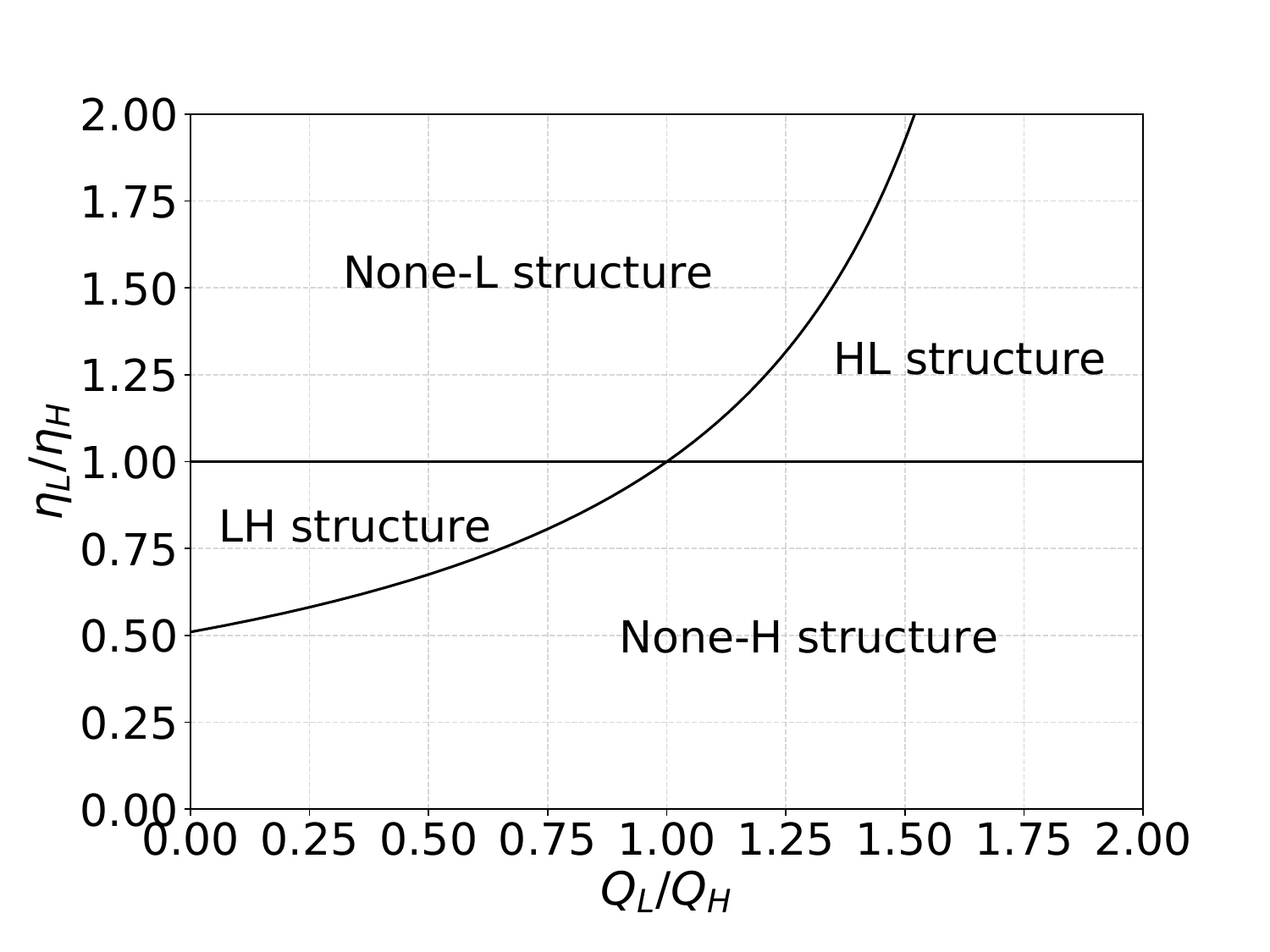}
    \caption{\rev{The distribution of the optimal policy structures.} \label{fig: dis of optimal}}
    \vspace{-0.4cm}
\end{figure}

\rev{Based on Theorem \ref{theorem: 1}, we can determine the structure of the optimal policy by calculating the success probability $Q_L, Q_H$ and the cost-effectiveness $\eta_L,\eta_H$} as shown in Fig.~\ref{fig: dis of optimal}. Here, for illustrative purposes, we set $Q_H=0.49$ and $\eta_H=5.71$ without loss of generality.

\rev{Next, we investigate the impact of the vehicles' parameters, such as arrival probabilities, on the thresholds of the optimal policy. By uncovering these connections, we reveal the underlying rationale for the optimal policy, which may apply to broader scenarios. The upcoming presentation will offer detailed analyses.}

\subsubsection{\rev{The impact of arrival probabilities on thresholds}} \rev{To examine this impact, we conducted simulations as shown in Fig.~\ref{fig: arrival prob}.} \rev{In Figs.~\ref{fig: p_L} and \ref{fig: p_H}, the x-axis represents arrival probabilities for \Tl~and \Th~vehicles, respectively, while the y-axis indicates the AoI. The thresholds $\delta_L$ and $\delta_H$ correspond to the AoI at which the company initiates recruitment of \Tl~and \Th~vehicles alone, respectively. For illustrative purposes, we use default parameters as follows: $\beta=0.0001, p_L=0.5, p_H=0.5, r_L=0.6, r_H=0.7, c_L=2, c_H=2.5$ without loss of generality. Note that we use the same default parameters in Fig.~\ref{fig: cost}.}

\begin{figure}[h]
    \centering
    \vspace{-0.4cm}
    \subcaptionbox{\label{fig: p_L}}{\includegraphics[width = .5\linewidth]{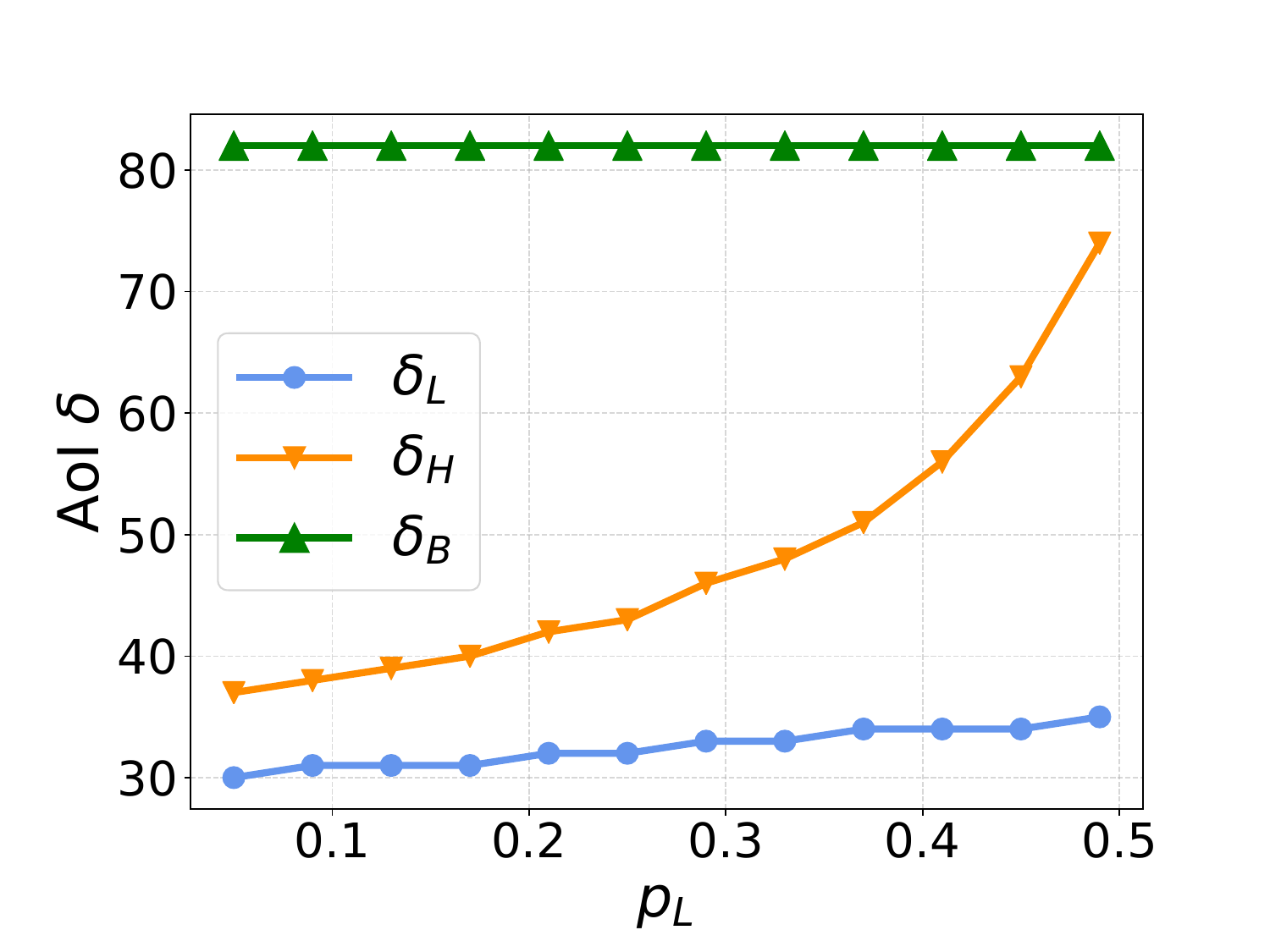}}\hfill
    \subcaptionbox{\label{fig: p_H}}{\includegraphics[width = .5\linewidth]{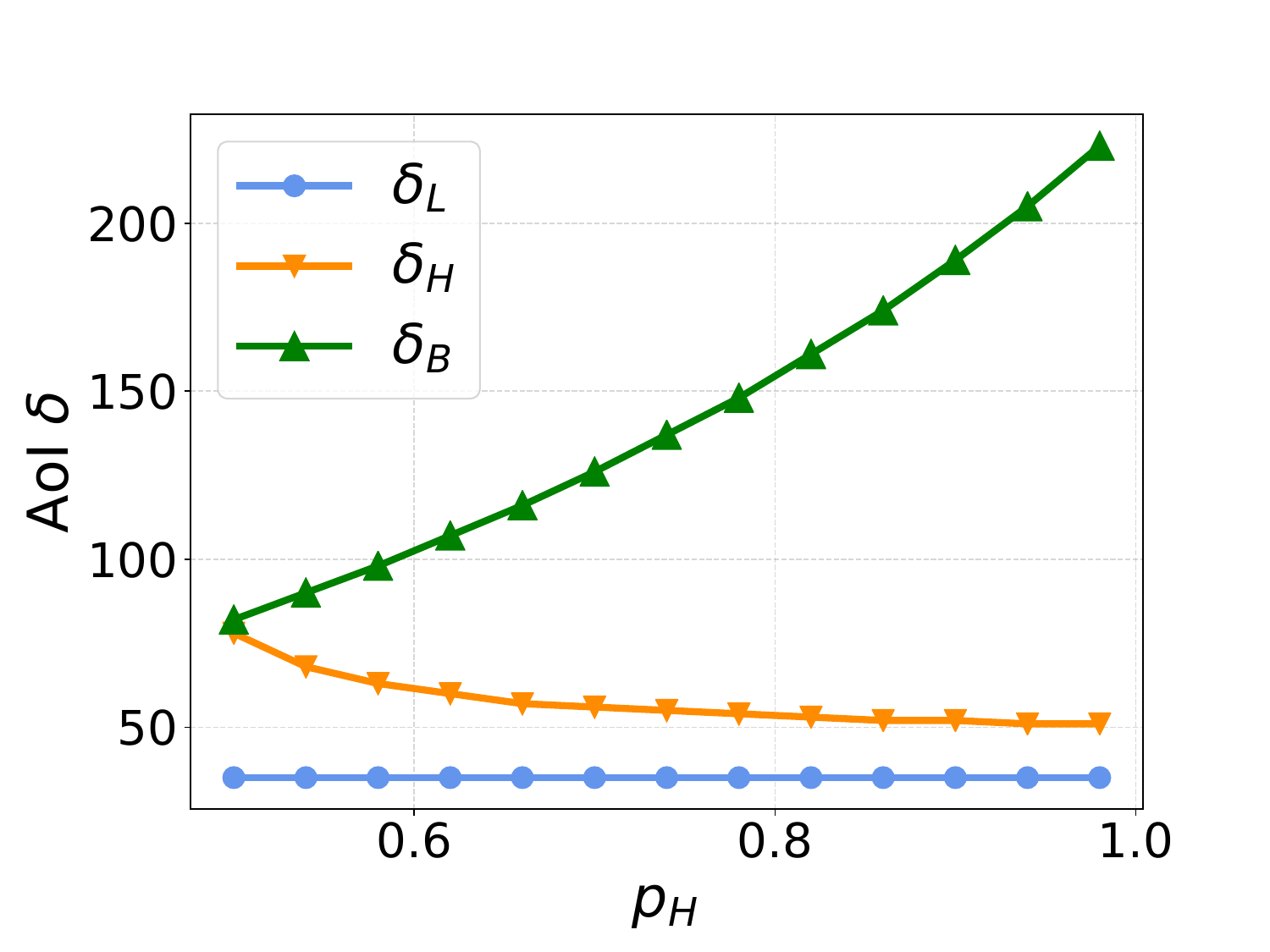}}
    \caption{The thresholds in LH structure vary as the arrival probability increases. Fig.~\ref{fig: p_L} and \ref{fig: p_H} manipulate the arrival probability $p_L$ and $p_H$ of \Tl~and \Th~vehicles, respectively.}
    \label{fig: arrival prob}
    \vspace{-0.4cm}
\end{figure}

\begin{obs}
    \rev{Fig.~\ref{fig: p_L} shows that in LH structure, $\delta_L$ increases as the arrival probability $p_L$ of \Tl~vehicles rises. Conversely, Fig.~\ref{fig: p_H} demonstrates that in LH structure, as the arrival probability $p_H$ of \Th~vehicles increases, the threshold $\delta_H$ decreases.}
\end{obs}

We first define the \emph{marginal cost-effectiveness} to help explain the above observation.
\begin{definition}
    \rev{The marginal cost-effectiveness of changing action from $a_1$ to $a_2$ is the ratio between the additional recruitment costs and the success probability increase.}
    \begin{equation}
        \gamma_{a_1,a_2} = \frac{p_{a_2}c_{a_2}-p_{a_1}c_{a_1}}{Q_{a_2}-Q_{a_1}}
    \end{equation}
\end{definition}

This metric assesses the cost-effectiveness of a company's change of action. A smaller value for marginal cost-effectiveness indicates that the company is more inclined to alter its course of action. For instance, if $Q_{a_2}>Q_{a_1}>0$ and $\gamma_{a_1,a_2}$ is very small, then the company is likely to opt for action $a_2$ instead of $a_1$.

The company may choose to delay recruitment when vehicles arrive frequently, as the threshold for initiating recruitment of \Tl~vehicles should increase with the corresponding arrival probability $p_L$, as depicted in Fig.~\ref{fig: p_L}. However, this pattern does not hold for the threshold $\delta_H$ and \Th~vehicles in LH structure. This is because the marginal cost-effectiveness $\gamma_{L, H}$ decreases as $p_H$ increases, which implies that changing to take action \Htype~from action \Ltype~becomes more cost-effective. Hence, as shown in Fig.~\ref{fig: p_H}, though success probability $Q_H$ increases in $p_H$, the company becomes more willing to change to take action \Htype.

\begin{insight}
    \rev{When vehicles of a certain type arrive more frequently, the company may consider initiating corresponding vehicle recruitment earlier. This should be done if the marginal cost-effectiveness of changing to the action decreases, \emph{i.e.,} recruiting vehicles of this type alone.}
\end{insight}

\subsubsection{The impact of operational costs on thresholds} We show this relationship in Fig.~\ref{fig: cost}, where the x-axis represents operational costs of \Tl~and \Th~vehicles (from left to right, respectively) and the y-axis indicates the AoI. Note that the threshold $\delta_B$ denotes the AoI at which the company starts recruiting both types of vehicles.

\begin{figure}[h]
    \centering
    \vspace{-0.4cm}
    \subcaptionbox{\label{fig: c_L}}{\includegraphics[width = .5\linewidth]{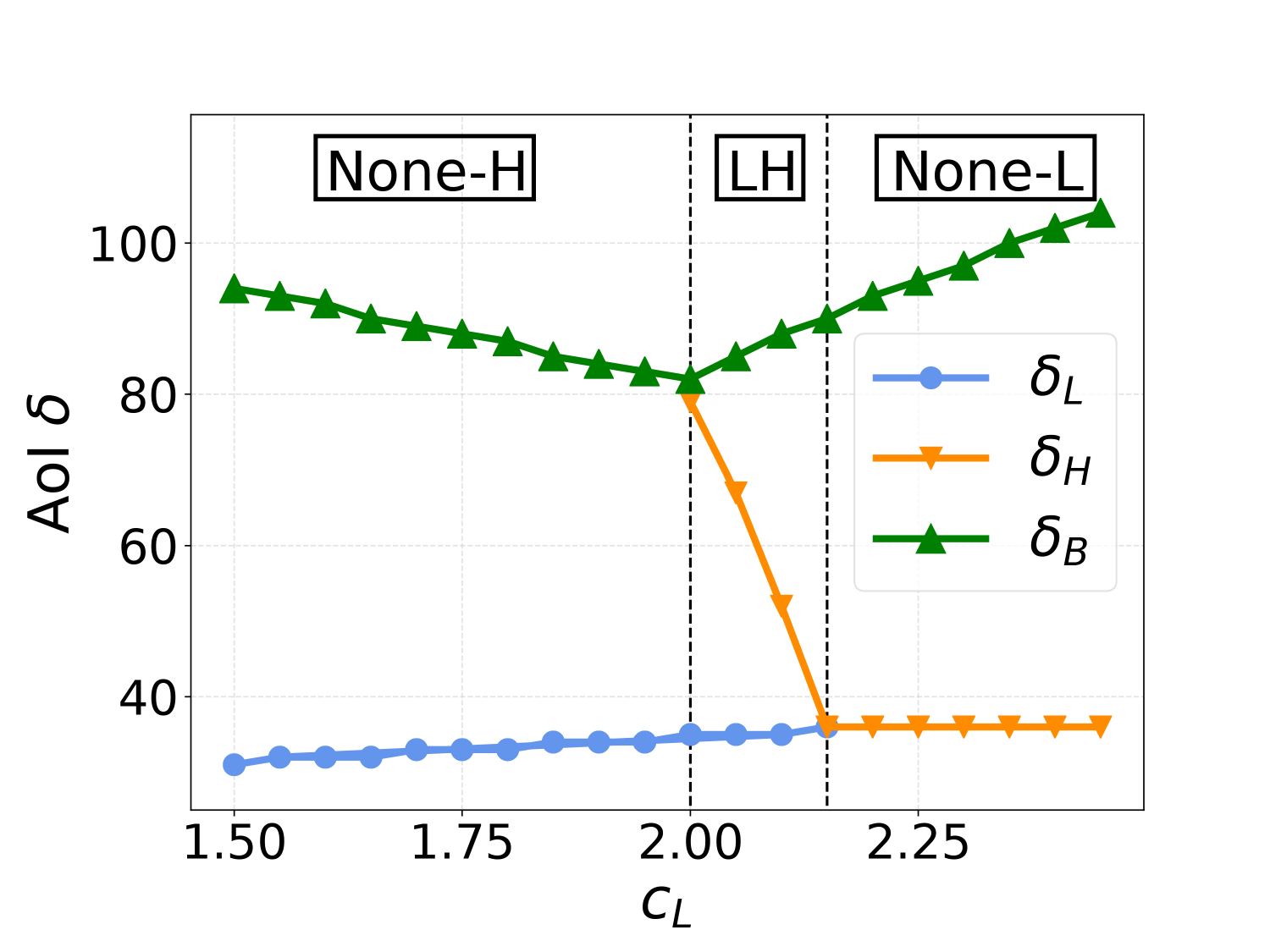}}\hfill
    \subcaptionbox{\label{fig: c_H}}{\includegraphics[width = .5\linewidth]{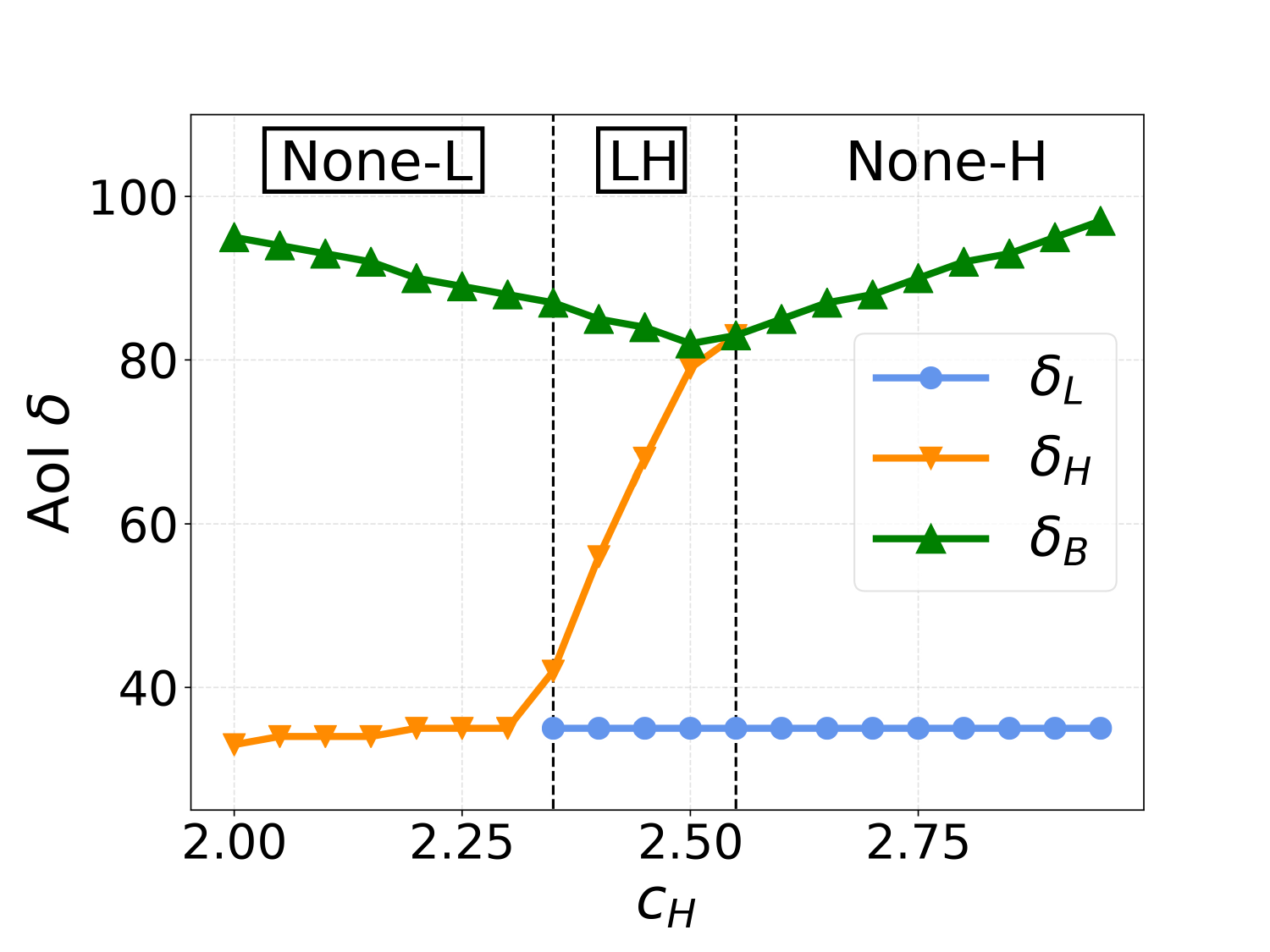}}
    \caption{The thresholds vary as the operational cost increases. Fig.~\ref{fig: c_L} and \ref{fig: c_H} depict the optimal policy transitioning between None-H to LH to None-L and None-L to LH to None-H structures, respectively.}
    \label{fig: cost}
    \vspace{-0.4cm}
\end{figure}

\begin{obs}
    Figs.~\ref{fig: c_L} and \ref{fig: c_H} show that as the operational cost of vehicles increases, the threshold $\delta_B$ of recruiting vehicles of both types will first decrease.
\end{obs}    

The intuition is that the company may tend to delay recruitment when faced with higher operational costs. For example, Fig.~\ref{fig: c_L} illustrates that $\delta_B$ of LH and None-L structures increases with the operational cost $c_L$ of \Tl~vehicles. However, $\delta_B$ of None-H structure decreases in $c_L$, which is counter-intuitive. This is because the cost-effectiveness difference in the action \Ltype~and \Both, represented by $\eta_L$ and $\eta_B$, decreases in $c_L$, and $\gamma_{LB}$ is independent of $c_L$. Therefore, the company is likely to opt for action \Both~instead of \Ltype. 

\begin{insight}
    The company may consider initiating recruitment of all vehicle types sooner when facing a higher operational cost. This should be done if the cost-effectiveness difference between the last two actions in the optimal recruitment order decreases.
\end{insight}

\rev{Sensing capabilities have a similar impact on thresholds as operational costs, as detailed in the online appendix\cite{appendix}. This yields the following insight:
\begin{insight}
    The company may consider initiating recruitment of all vehicle types sooner when facing a higher sensing capability. This should be done if the marginal cost-effectiveness of changing to action \Both~decreases.
\end{insight}}

\subsection{Bound-based Relative Value Iteration Algorithm} \label{subsec: BRVIA}

In this subsection, we propose the BRVI algorithm that computes the optimal policy efficiently.

\rev{The RVI algorithm\cite{white1963dynamic} is a classical approach for solving the average-cost MDP, which has been widely studied and improved upon \cite{hsu2019scheduling}. However, updating an infinite number of states in each iteration makes this algorithm impractical to solve our problem. To circumvent this issue, we propose using a sequence of finite-state approximate MDPs to effectively approximate the original MDP\cite{hsu2019scheduling}.}

First, let $\delta^{(m)}(t)$ be the truncated AoI in time slot $t$,
\begin{equation}
    \delta^{(m)}(t+1) = \left\{\begin{array}{ll}
        1,  & \text{the HD map is updated,} \\
        \left[\delta(t)+1\right]^+_m, & \text{otherwise},
    \end{array}
    \right.
\end{equation}
where we choose $m$ to be the number of truncated states and define the notation $[x]_m^+$ by $[x]_m^+=x$ if $x<m$ and $[x]_m^+=m$ otherwise. Then we can define the approximate MDPs $\Lambda^{(m)}$, which is the similar as the original MDP except
\begin{itemize}
    \item \textbf{States}: $S(t)=\delta^{(m)}(t)\in\mathcal{S}^{(m)}$, \emph{i.e.,} $\delta^{(m)}$ is the truncated AoI at that slot. Let $\mathcal{S}^{(m)}\in\{n\leq m|n\in\Nsetp\}$ be the state space, which is countably finite.
        
    \item \textbf{State transition probability}: $P^{(m)}_{ss'}(a)$, \emph{i.e.,} the state transition probability from state $s$ to $s'$ under the action $a$. Suppose the state $s$ is $\delta^{(m)}(t)$, the non-zero $P^{(m)}_{ss'}(a)$ is as follows\\
    \begin{equation} 
    P^{(m)}_{ss'}(a)=\left\{
    \begin{array}{ll}
        1,&s'=[\delta^{(m)}(t)+1]_m^+,a=N\\            
        Q_a, &s'=1,a\not=N\\
        1-Q_a,&s'=[\delta^{(m)}(t)+1]_m^+,a\not=N\\
    \end{array}\right.
    \end{equation}
\end{itemize}


\rev{Though we can use the traditional RVI algorithm to derive the optimal policy by the approximate MDP $\Lambda^{(m)}$\cite{hsu2019scheduling}, it is time-consuming. In the traditional RVI algorithm, the most time-consuming step is the minimization operation in Eq.~\eqref{eq: optimal policy update} for all states in each iteration.
\begin{equation}\label{eq: optimal policy update}
    J_{n+1}(\delta) = \min_{a\in\mathcal{A}} u(\delta,a) + \Ex[J_{n}(\delta')] - J_{n}(1),
\end{equation}
where $\Ex[J_{n}(\delta')]=\sum_{\delta'\in\mathcal{S}^{(m)}} P^{(m)}_{\delta\delta'}(a) J_{n}(\delta')$. In each iteration round, the computational complexity of Eq.~\eqref{eq: optimal policy update} for all truncated states is $O(|\mathcal{A}|\times m)$, which results from the size of the feasible set, that is the action space $\mathcal{A}$, and the number of truncated states $m$. As the number of truncated states $m$ goes to infinity and the size of the feasible set is factorially increasing in the number of vehicle types, we can see that the traditional RVI algorithm is time-consuming.}

\rev{Structural RVI algorithm\cite{hsu2019scheduling} is an improved RVI algorithm that computes the optimal policy efficiently utilizing the threshold-type structural property to reduce the execution number of Eq.~\eqref{eq: optimal policy update} in each iteration. To develop a lower-complexity algorithm, we additionally optimize the feasible set to further accelerate the iteration based on the upper bounds of optimal thresholds. }

\rev{For illustrative purposes, we take LH structure for example. Other structures are provided in the online appendix\cite{appendix}.}

To reduce the feasible set $\mathcal{A}$, we first characterize the upper bounds of thresholds in LH structure. First, we define $[x]^+$ as the minimum integer larger than $x$.
\begin{lemma} \label{lemma: structure (1)}
    Consider the LH structure of the optimal policy in Theorem~\ref{theorem: 1}, the optimal thresholds $\delta_L$, $\delta_H$, $\delta_B$ are upper bounded by the following values, respectively. 
    \begin{equation}
        \hat{\delta}_L=\left[\sqrt{\frac{(1-\beta)\eta_L}{\beta}}\right]^+
    \end{equation}
    \begin{equation}
        \hat{\delta}_H=\left[\sqrt{\frac{(1-\beta)(\eta_HQ_H-\eta_LQ_L)}{\beta(Q_H-Q_L)}}\right]^+
    \end{equation}
    \begin{equation}
        \hat{\delta}_B=\left[\sqrt{\frac{(1-\beta)\eta_L}{\beta(1-Q_H)}}\right]^+
    \end{equation}
\end{lemma}

Next, based on Lemma~\ref{lemma: structure (1)}, we reduce the feasible set $\mathcal{A}$ in Eq.~\eqref{eq: optimal policy update}. In more detail, suppose that the AoI $\delta=\delta^{(m)}(t)$, we can define the reduced feasible set $\mathcal{A}_{LH}(\delta)$, as follows:
\begin{equation}\label{eq: feasible action space function} 
\mathcal{A}_{LH}(\delta) = \left\{\begin{array}{ll}
    \mathcal{A},  & \delta<\hat{\delta}_L \\
    \{L,H,B\}, & \hat{\delta}_L\leq\delta<\hat{\delta}_H\\
    \{H,B\}, & \hat{\delta}_H\leq\delta<\hat{\delta}_B\\
    \{B\}, & \delta\geq\hat{\delta}_B\\
\end{array}
\right.
\end{equation}

\rev{Finally, we use the reduced feasible set $\mathcal{A}_{LH}(\delta)$ to replace $\mathcal{A}$ in Eq.~\eqref{eq: optimal policy update} to accelerate this optimization operation.

We then leverage the structural property of the optimal policy defined in Theorem~\ref{theorem: 1}\footnote{The optimal policy for the truncated MDPs is threshold-type as well, according to the same proof as Theorem~\ref{theorem: 1}} to reduce the number of optimization operations in Eq.~\eqref{eq: optimal policy update}. Specifically, in each round, we directly set \Both~as the optimal action of the current state if it was taken as the optimal action of the previous state. Algorithm~\ref{alg: BRVIA} shows how the BRVI algorithm determines the optimal policy together with the reduced feasible set. }

\begin{algorithm}[!t]
    \caption{BRVI algorithm}
    \label{alg: BRVIA}
    \KwIn{the weighted factor: $\beta$\\
    \qquad \quad\Tl~vehicles' parameters: $p_L,c_L,r_L$\\
    \qquad \quad\Th~vehicles' parameters: $p_H,c_H,r_H$\\
    }
    \KwOut{the optimal policy $\phi^*$}
    Determine the structure of the optimal policy based on Theorem \ref{theorem: 1}.
    \BlankLine
    \tcc{Take LH structure as an example}
    Set the reduced feasible set $\mathcal{A}_{LH}(\delta)$ based on Eq.~\eqref{eq: feasible action space function} for all truncated states $\delta\in\mathcal{S}^{(m)}$.
    \BlankLine
    Initialize $J^{(m)}(\delta)\leftarrow0$ for all truncated states $\delta\in\mathcal{S}^{(m)}$. \\
    Record $J^{(m)}$ of the last round $J^{(m)}_{last}(\delta) \leftarrow 0$.
    \BlankLine
    \While{$\left|\frac{J^{(m)}-J^{(m)}_{last}}{J^{(m)}}\right|>\theta$}{$J^{(m)}_{last}(\delta)\leftarrow J^{(m)}(\delta)$ for all $\delta\in\mathcal{S}^{(m)}$.
    \BlankLine
        \For{$\delta\in\mathcal{S}^{(m)}$}{
        \eIf{there exists $y>0$ such that $\phi^*(\delta-y)=B$}{$\phi^*(\delta) \leftarrow B$.}{
        $\phi^*(\delta) \leftarrow \arg\min_{a\in\mathcal{A}_{LH}(\delta)} u(\delta,a)+\Ex[J^{(m)}(\delta')]$.}
        $J^{(m)}_{tmp}(\delta)\leftarrow u(\delta,\phi^*(\delta))+\mathbb{E}[J^{(m)}(\delta')]-J^{(m)}(1)$.
        }
        $J^{(m)}(\delta)\leftarrow J^{(m)}_{tmp}(\delta)$ for all $\delta\in\mathcal{S}^{(m)}$.   
    }
\end{algorithm}

Finally, we establish the optimality of the BRVI algorithm for the original MDP $\Lambda$.

\begin{theorem} \label{thm: alg converges}
    \rev{As the number of truncated state $m$ goes infinity,} the output $\phi^*(\delta)$ in Algorithm~\ref{alg: BRVIA} converges to the optimal policy of the original MDP $\Lambda$ in a finite number of iterations.
\end{theorem}
Proofs are provided in the online appendix\cite{appendix}.

 
\section{Simulation Results}\label{Sec: simulations}
In this section, we conduct extensive simulations for the optimal policy and the proposed BRVI algorithm. First, we show the performance of the optimal policy in Section~\ref{subsec: comparison of the average cost under different policies}. Next, Section~\ref{subsec: efficiency comparison} shows the efficiency of the proposed BRVI algorithm.

\subsection{Performance of the optimal policy} \label{subsec: comparison of the average cost under different policies}
In the subsection, we compare the optimal policy defined in Theorem~\ref{theorem: 1} and the following two baseline policies in terms of the company's average cost.
\begin{itemize}
    \item Zero-wait policy: a naive policy without considering the tradeoff between AoI loss and recruitment costs, \emph{i.e.,} recruiting vehicles of both types at all times.
    \item Dynamic pricing\cite{wang2022dynamic}: a policy considering the tradeoff between AoI loss and recruitment costs but ignoring crowdsource vehicles' heterogeneous sensing capability.
\end{itemize}

\begin{figure}[htbp]
    \centering
    \vspace{-0.4cm}
    \includegraphics[width=6cm]{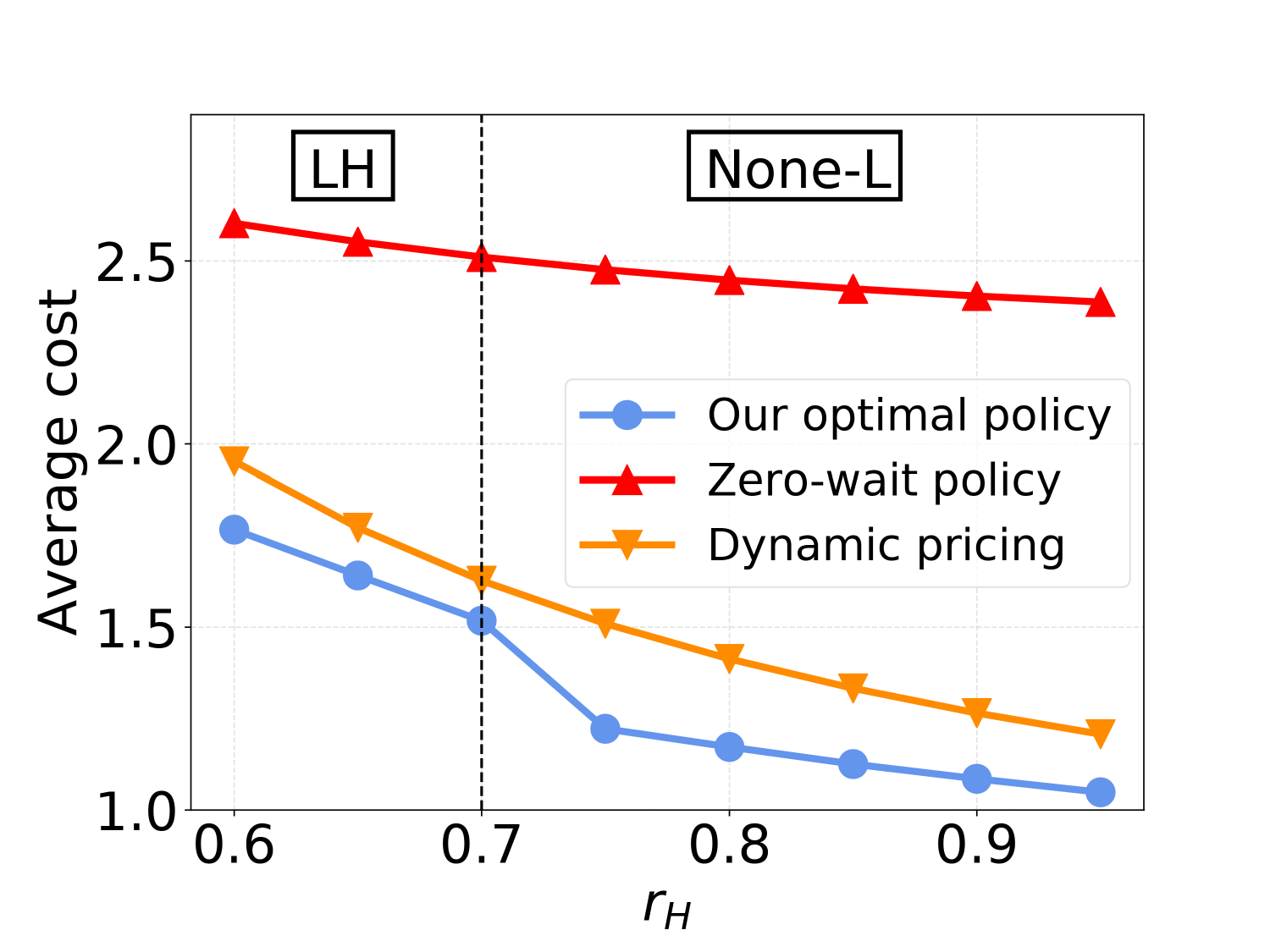}
    \caption{Comparison of average cost under different policies.}
    \label{fig: average cost comparison}
\end{figure}

\rev{To show the outperformance of our optimal policy, we conduct the simulation as shown in Fig.~\ref{fig: average cost comparison}, where the x-axis represents the qualified ratio $\Hqual$ of \Th~vehicles while the y-axis indicates the company's average cost. For illustrative purposes, we set other parameters as follows: the weighted factor $\beta=0.3$, \Tl~and \Th~vehicles' arrival probabilities $p_L=0.5,p_H=0.95$, \Tl~vehicles' sensing capabilities $r_L=0.6$, \Tl~and \Th~vehicles' operational costs $c_L=2,c_H=2.5$.}

\rev{Fig.~\ref{fig: average cost comparison} shows that our optimal policy significantly decreases the company's average cost. As shown in Fig.~\ref{fig: average cost comparison}, the optimal policy reduces the average cost by $46.84\%$ on average compared to the zero-wait policy and up to $19.04\%$ compared to dynamic pricing. Notice that the sharp change in the blue curve is due to the structural change in the optimal policy, \emph{i.e.,} from LH to None-L structure, while dynamic pricing remains unchanged.}

\subsection{Efficiency Comparison} \label{subsec: efficiency comparison}
In the subsection, we compare the proposed BRVI algorithm to the following two baseline algorithms.
\begin{itemize}
    \item RIV algorithm: a classical dynamic programming method to solve the average-cost MDP problem.
    \item Structural RVI algorithm: an improved RVI algorithm to reduce computational complexity based on the threshold-type structural property of the optimal policy, which is also state-of-the-art.
\end{itemize}

\rev{We perform simulations in Table~\ref{tab: eff} to compare the convergence time for different values of the weighted factor $\beta$. In Table~\ref{tab: eff}, the convergence time is measured in minutes, and numbers in parentheses indicate the improvement of our proposed BRVI algorithm compared to RVI and structural RVI algorithms, respectively. Fur illustrative purposes, we set other parameters as $p_L=0.5,p_H=0.95,r_L=0.6,r_H=0.7,c_L=2,c_H=2.5$, the convergence tolerance $\theta=10^{-10}$ and the number of truncated state $m=1000$.}

\begin{table}[!t] 
    \centering 
    \caption{Convergence time (min) comparison\label{tab: eff}}
    \begin{tabular}{|c|ccc|}
    \hline
    \diagbox[innerwidth=1.2cm]{$\beta$}{Alg.} & RVI & Structural RVI & \textbf{BRVI(proposed)} \\
    \hline    
    $10^{-4}$  & $10.84$ & $3.99$ & \bm{$3.66$}(\bm{$(66.19\%,8.20\%)$})  \\
    $10^{-3}$ & $2.32$  & $0.84$ & \bm{$0.77$}(\bm{$(66.65\%,7.49\%)$})   \\
    $10^{-2}$ & $0.86$ & $0.33$   & \bm{$0.27$}(\bm{$(68.11\%,17.49\%)$})   \\
    $10^{-1}$ & $0.22$  & $0.09$   & \bm{$0.07$}($\bm{68.86\%,21.45\%}$) \\
    \hline
    \end{tabular}
    \vspace{-0.4cm}
\end{table}

Table~\ref{tab: eff} demonstrates that, on average, the proposed BRVI algorithm reduces convergence time by $13.66\%$ compared to the Structural RVI algorithm and by even $67.4\%$ compared to the RVI algorithm. Thus, we can conclude that the BRVI algorithm greatly reduces the convergence time.

\rev{The complexity of the structural RVI algorithm increases with the size of the feasible set, which expands factorially with the number of vehicle types. In practice, there are more than two types of crowdsource vehicles. Thereby, our proposed algorithm can significantly reduce the complexity even in such cases.}

\section{Conclusion and Future work}
\rev{This paper studies the tradeoff between freshness and recruitment costs involving crowdsource vehicles that arrive randomly and have heterogeneous operational costs and sensing capabilities.} To achieve a better tradeoff, we define the average cost and formulate the average minimization problem as an MDP $\Lambda$. \rev{Leveraging MDP, we prove that the optimal policy is threshold-type age-dependent. Additionally, our findings reveal some counter-intuitive insights. In some cases, the company should initiate vehicle recruitment earlier when vehicles arrive more frequently, or have higher operational costs or sensing capabilities.} The optimal policy reduces the average cost by $19.04\%$ compared to the state-of-the-art mechanism. Our proposed BRVI algorithm also reduces the convergence time by an average of $13.66\%$ compared to the existing algorithm. 

In the future work, we could extend these results to a continuous-time system with more general distributions of operational costs and sensing capabilities, which is more practical. Additionally, a more general setting could be explored, such as competition among multiple companies, which could be formulated using a duopoly model.



\bibliographystyle{IEEEtran}
\bibliography{IEEEabrv,ref}

\newpage
\section{Appendices: The impact of the Vehicles’ Parameters}
In this section, we illustrate the impact of sensing capabilities on thresholds in Fig.~\ref{fig: capability}, where the x-axis represents sensing capabilities of \Tl~and \Th~vehicles (from left to right, respectively) and the y-axis indicates the AoI. Note that the threshold $\delta_B$ denotes the AoI at which the company starts recruiting both types of vehicles.
\begin{figure}[h]
    \centering
    \subcaptionbox{\label{fig: r_L_}}{\includegraphics[width = .5\linewidth]{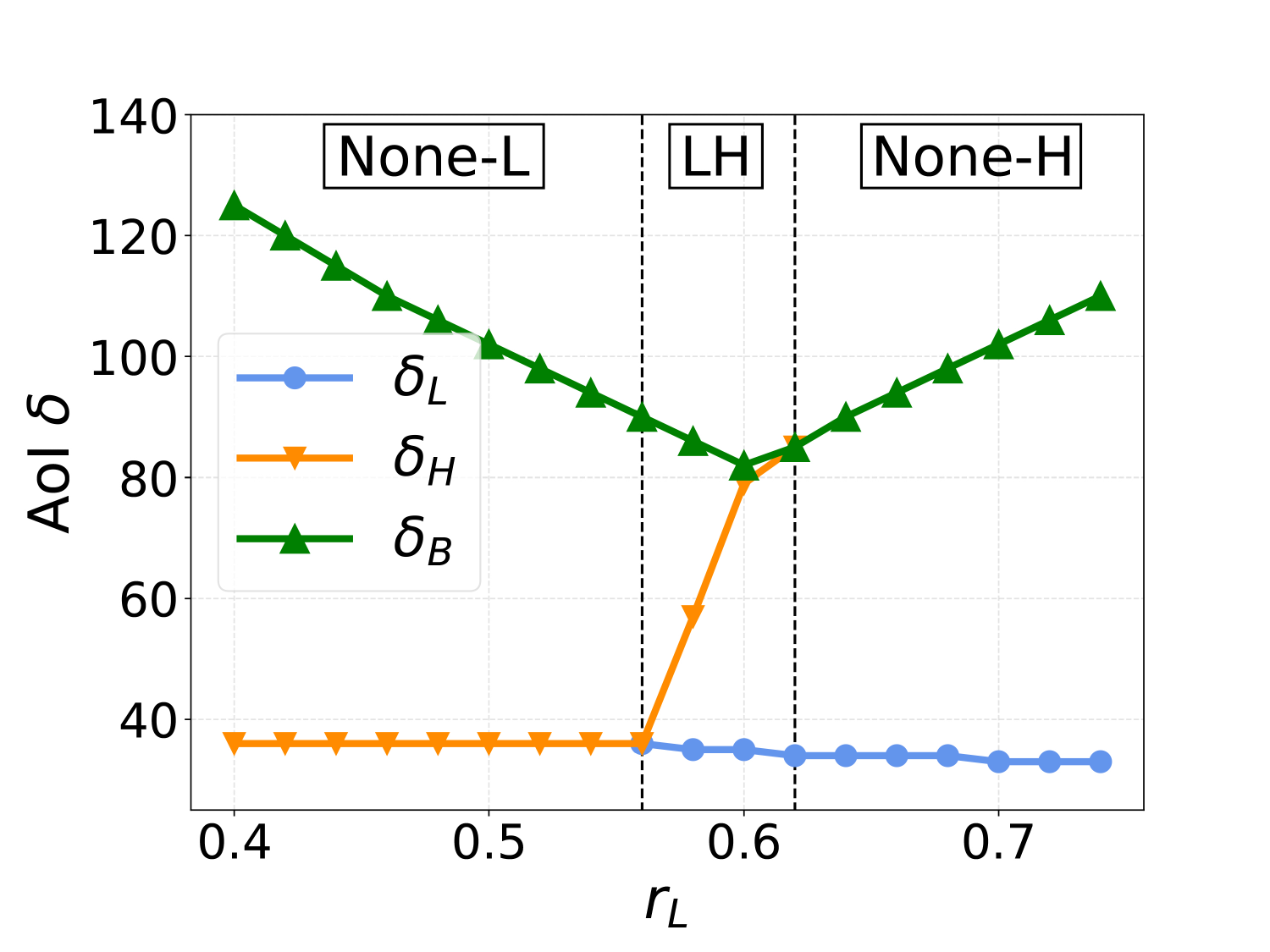}}\hfill
    \subcaptionbox{\label{fig: r_H_}}{\includegraphics[width = .5\linewidth]{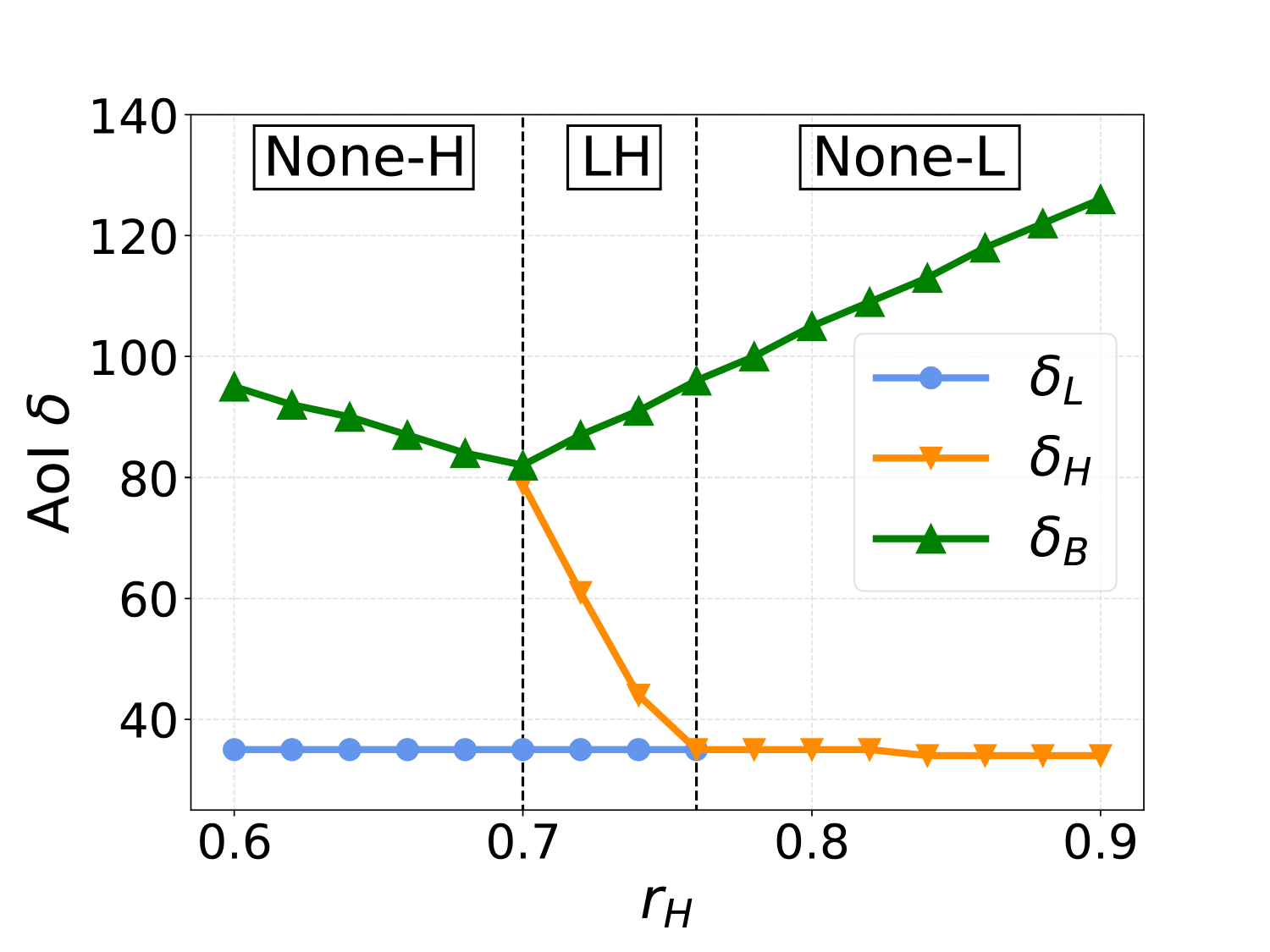}}
    \caption{The thresholds vary as the operational cost increases. Fig.~\ref{fig: r_L_} and \ref{fig: r_H_} depict the optimal policy transitioning between None-L to LH to None-H structure and  None-H to LH to None-L structure, respectively.}
    \label{fig: capability}
\end{figure}

\begin{obs}
    Fig.~\ref{fig: r_L_} and \ref{fig: r_H_} show that as vehicles' sensing capability increases, the threshold $\delta_B$ of recruiting vehicles of both types will first decrease. 
\end{obs}

The intuition is that the company may tend to delay recruitment when facing vehicles with higher sensing capabilities. For instance, Fig.~\ref{fig: r_L_} illustrates that $\delta_B$ in LH and None-H structures increases with the sensing capability $\Lqual$ of \Tl~vehicles. However, $\delta_B$ in None-L structure decreases in $\Lqual$, which is counter-intuitive. This is because the cost-effectiveness difference in the action \Both~and \Htype, represented by $\ceb$ and $\ceh$, decreases in $\Lqual$, and $\gamma_{HB}$ also decreases in $\Lqual$. 
\begin{equation}
    \gamma_{HB} = \frac{c_L}{r_L(1-Q_H)} \text{ decreases in } r_L 
\end{equation}

Therefore, the company is likely to opt for action \Both~instead of \Ltype.

\section{Appendices: Proofs of Main Results}
\subsection{Preliminaries}
To characterize the optimal policy of the MDP $\Lambda$, we start with introducing a $\dis$-discounted MDP $\Lambda_{\dis}$, which has the same definition as the MDP $\Lambda$ except for the objective function. It aims at minimizing the expected total $\dis$-discount cost $J^{\dis}(\delta,\phi)$ given an initial state\footnote{The state of any MDPs in this paper is the AoI, therefore, we just use $\delta$ to represent the state for convenience hereafter.} $s=\delta$ and under the policy $\phi$. 
\begin{equation}\label{eq. discount MDP average cost}
    J^{\dis}(\delta,\phi) = \lim_{T\to\infty} \sum_{t=1}^T \dis^t \cdot u(S(t),A^{\phi}(t)|S(1)=\delta),
\end{equation}
where $\rho\in(0,1)$ is the discount factor. 

To distinguish the optimal policies of MDPs $\Lambda$ and $\Lambda_{\rho}$, we use $\Lambda$-$optimal$ and $\Lambda_{\dis}$-$optimal$ to denote them respectively.


The minimum value function $J^{\dis,*}(\delta)$ satisfies the following property.
\begin{lemma}\label{lemma: bounded value function}
    For any given $\delta$ and $\dis$, $J^{\dis,*}(\delta)<\infty$.
\end{lemma}
\begin{proof}
    See Appendix A.
\end{proof}

According to \cite{sennott1989average} and Lemma \ref{lemma: bounded value function}, we have a direct lemma.
\begin{lemma} \label{lemma: discounted-MDP}
    \begin{enumerate}
        \item[(a)] The value function $J^{\rho}(\delta)$ satisfies the Bellman equation
        \begin{equation} \label{eq: bellman equation for beta}
            J^{\dis,*}(\delta) = \min_{a\in\mathcal{A}} u(\delta,a) + \dis \sum_{\delta'\in\mathcal{S}}P_{\delta\delta'}(a)J^{\dis,*}(\delta').
        \end{equation}
        \item[(b)] There exists a deterministic stationary policy $\phi^{\dis,*}$ that satisfies the Bellman equation \eqref{eq: bellman equation for beta}.
        \item[(c)] We can define a value iteration $J_t^\dis(s)$ by $J_0^\dis(s)=0$ and for any $t>0$,
        \begin{equation}
            J^\dis_{t+1}(\delta) = \min_{a\in\mathcal{A}} u(\delta,a) + \dis \sum_{\delta'\in\mathcal{S}}P_{\delta\delta'}(a)J^\dis(\delta').
        \end{equation}
        Then $J^\dis_{t}(\delta) \to J^{\dis,*}(\delta)$ as $t\to\infty$, for any $\delta$ and $\dis$.
    \end{enumerate}
\end{lemma}

Then we can leverage Lemma~\ref{lemma: discounted-MDP}(c) to prove.
\begin{lemma} \label{lemma: the monotonicity value function}
    For any $\dis$, $J^{\dis,*}(\delta)$ is non-decreasing in $\delta$.
\end{lemma}
\begin{proof}
    See Appendix B.
\end{proof}

Based on Lemma~\ref{lemma: discounted-MDP} and Lemma~\ref{lemma: the monotonicity value function}, the following lemma gives the connection between $\Lambda$-$optimal$ and $\Lambda_{\dis}$-$optimal$ policies.
\begin{lemma}\label{lemma: from discount to average}
    \begin{enumerate}
        \item[(a)] There exists a determinsitic stationary policy $\phi^*$ that is $\Lambda$-$optimal$.
        \item[(b)] There exists a finite constant $J^*=\lim_{\rho\to1}(1-\rho) J^{\dis,*}(\delta)$ for every state $\delta$ such that the minimum average cost is $J^*$, independent of initial state $\delta$. 
        \item[(c)] For any sequence $\{\rho_n\}_{n=1}^\infty$ of discount factors that converges to 1, there exists a subsequence $\{\theta_n\}_{n=1}^\infty$ such that $\lim_{n\to\infty} \phi^{\theta_n,*} = \phi^*$, where $\phi^*$ is the $\Lambda$-$optimal$ policy. 
    \end{enumerate}     
\end{lemma}
\begin{proof}
    See Appendix C.
\end{proof}

\subsection{Proof of Theorem \ref{theorem: 1}}
To show Theorem \ref{theorem: 1}, we start by providing a similar Theorem~\ref{thm: 2} for the $\dis$-discount MDP $\Lambda_{\dis}$. Based on Lemma~\ref{lemma: from discount to average}, we prove that Theorem \ref{theorem: 1} holds.

Before providing a similar theorem, we prepare some lemmas to help the proof.
\begin{lemma}\label{lemma: action N appears first and action B appears finally}
    For any $\beta$, we have
    \begin{itemize}
        \item If $\phi^{\beta}(1)\not=N$, then $\phi^{\beta}(\delta)\not=N$ for all $\delta$.
        \item If $\phi^{\beta}(\delta)=B$, then $\phi^{\beta}(\delta')=B$ for all $\delta'>\delta$.
    \end{itemize}
\end{lemma}
\begin{proof}
    See Appendix D.
\end{proof} 
        
\begin{lemma}\label{lemma: cost efficiency inequality}
There are some inequalities:
    \begin{equation} \label{eq: expected recruitment expenses per HD map}
        \ceb \geq \min\left\{\cel,\ceh\right\}
    \end{equation} 
    \begin{equation} \label{eq: eta1}
        \frac{p_Hc_H-p_Lc_L}{Q_H-Q_L}<\cel \text{ when } Q_L>Q_H, \cel<\ceh
    \end{equation}
    \begin{equation} \label{eq: eta2}
        \frac{p_Hc_H-p_Lc_L}{Q_H-Q_L}<\ceh \text{ when } Q_H>Q_L,\cel>\ceh
    \end{equation}
\end{lemma}
\begin{proof}
    See Appendix E.
\end{proof} 

\begin{theorem} \label{thm: 2}
    There exists a threshold-type age-dependent stationary deterministic optimal policy $\phi^{\dis,*}$ of the MDP $\Lambda_\dis$. Depending on the parameter setting, such an optimal policy falls into one of the following four structures:
    \begin{enumerate}
        \item \textbf{LH structure} (in Fig.~\ref{fig: LH structure}): if $\frac{Q_L}{Q_H} \leq 1$ and $\frac{1-Q_H}{1-Q_L}<\frac{\eta_L}{\eta_H}<1$, the optimal policy includes actions in the order of \None, \Ltype, \Htype, \Both, with three thresholds satisfying $1 \leq \delta_L \leq \delta_H \leq \delta_{B}\in\Nsetp$.
        \item \textbf{HL structure} (in Fig.~\ref{fig: HL structure}): if $\frac{Q_L}{Q_H} > 1$ and $1<\frac{\eta_L}{\eta_H}<\frac{1-Q_H}{1-Q_L}$, the optimal policy includes actions in the order of \None, \Htype, \Ltype, \Both, with three thresholds satisfying $1\leq \delta_H \leq \delta_L \leq \delta_{B}\in \Nsetp$.        
        \item \textbf{None-L structure} (in Fig.~\ref{fig: None-L structure}): if $\frac{\eta_L}{\eta_H} \geq \max\{1,\frac{1-Q_H}{1-Q_L}\}$, the optimal policy includes actions in the order of \None, \Htype, \Both, with two thresholds satisfying $1 \leq \delta_H \leq \delta_{B} \in \Nsetp$.        
        \item \textbf{None-H structure} (in Fig.~\ref{fig: None-H structure}): if $\frac{\eta_L}{\eta_H}<\min\{1,\frac{1-Q_H}{1-Q_L}\}$, the optimal policyincludes actions in the order of \None, \Ltype, \Both, with two thresholds satisfying $1 \leq \delta_L \leq \delta_{B} \in \Nsetp$.
    \end{enumerate}
\end{theorem}
\begin{proof}
    Before we start the proof, there are some useful definitions.
    \begin{definition} \label{def: Q function}
        The candidate value function $Q^\dis(\delta, a)$ under the state $\delta$ and the action $a$ is as follows.
        \begin{equation}
            \begin{split}
                Q^\dis(\delta,a)=&(1-\beta)p_ac_a + \beta(1-Q_a) \ageloss\\
                &+ \dis\left(Q_a J^{\dis,*}(1) + (1-Q_a) J^{\dis,*}(\delta+1)\right)
            \end{split}    
        \end{equation}
    \end{definition}
    By this notation, the minimum average cost under the initial state $\delta$ and the optimal action of this state can be written as
    \begin{equation} 
        J^{\dis,*}(\delta) = \min_{a\in\mathcal{A}} Q^\dis(\delta,a)
    \end{equation}
    \begin{equation}    
        \phi^{\dis,*}(\delta) = \arg\min_{a\in\mathcal{A}} Q^\dis(\delta,a) 
    \end{equation}
    
    \begin{definition}\label{def: D function}
        The difference in candidate value functions under the state $\delta$ and different actions $a_1,a_2$ by $D^\dis(\delta,a_1,a_2)$ is as follows.
        \begin{equation}
            \begin{split}
                D^\dis(\delta,a_1,a_2)=Q^\dis(\delta,a_1)-Q^\dis(\delta,a_2)
            \end{split}    
        \end{equation}
    \end{definition}
    Notice that $D^\dis(\delta,a_1,a_2)<0$ implies that the company prefers action $a_1$ to $a_2$ under the state $\delta$.
    
    Based on Lemma.~\ref{lemma: action N appears first and action B appears finally}, there are two cases: $\phi^{\dis,*}(1)=N$ and $\phi^{\dis,*}(1)\not=N$. We only need to consider the first case as the second case is a special case of the first in this theorem.
    
    Note that $\phi^{\beta,*}(1)=N$. We can determine the next action in the order by the following functions.
    \begin{itemize}
        \item $D^\dis(\delta,N,L) = Q_L\left(F(\delta) -(1-\beta)\cel\right)$
        \item $D^\dis(\delta,N,H) = Q_H\left(F(\delta) -(1-\beta)\ceh\right)$
        \item $D^\dis(\delta,N,B) = Q_B\left(F(\delta) -(1-\beta)\ceb\right)$
    \end{itemize}
    where $F(\delta) = \dis h_\dis(\delta+1) + \beta \ageloss = \dis (J^\dis(\delta+1)-J^\dis(1)) + \beta \ageloss$, which increases in $\delta$. 
    
    Because $D^\dis(\delta, N, a)$ for $a=L, H, B$ are increasing in $\delta$, and are possibly negative first and definitely positive when $\delta$ is large enough. Moreover, as Eq.~\eqref{eq: expected recruitment expenses per HD map} in Lemma \ref{lemma: cost efficiency inequality}, $D^\beta(\delta, N, L)$ or $D^\beta(\delta, N, H)$ first turns to be positive, which implies that the next action of the optimal policy is either \Ltype~or \Htype~instead of \Both. More specifically, when $\cel<\ceh$, the next action is \Ltype, otherwise, the next action is \Htype.
    \begin{enumerate}
        \item[(a)] Case 1: $\cel<\ceh$. Since the second action is \Ltype, we can consider the following two functions.
        \begin{equation*}
        \begin{split}
            &D^\dis(\delta,L,H)=(Q_H-Q_L)
            \left(F(\delta) - (1-\beta)\frac{p_Hc_H-p_Lc_L}{Q_H-Q_L}\right)
        \end{split}        
        \end{equation*}
        \begin{equation*}
        \begin{split}
            &D^\dis(\delta,L,B)=(Q_B-Q_L)\left(F(\delta) - (1-\beta)\frac{p_Hc_H}{Q_B-Q_L}\right)
        \end{split}        
        \end{equation*}
        Then, there are three possibilities:
        \begin{itemize}
            \item If $Q_H<Q_L$, then the next action is \Both, which is exactly None-H structure. This is because there does not exist a $\delta$ satisfying that $D^\dis(\delta,L,H)>0$ and $D^\dis(\delta,N,L)>0$ when $Q_H<Q_L$ and $\cel<\ceh$. If otherwise, we have 
            \begin{eqnarray*}
                D^\dis(\delta,L,H)>0 \Rightarrow F(\delta) < (1-\beta)\frac{p_Hc_H-p_Lc_L}{Q_H-Q_L}
            \end{eqnarray*}
            \begin{eqnarray*}
                D^\dis(\delta,N,L)>0 \Rightarrow F(\delta) > (1-\beta) \cel
            \end{eqnarray*}
            According to Eq.~\eqref{eq: eta1}, they are a contradiction.
            
            \item If $Q_H>Q_L$ and $\frac{p_Hc_H-p_Lc_L}{Q_H-Q_L}>\frac{p_Hc_H}{Q_B-Q_L}$, the next action is \Both, which is exactly None-H structure. This is because as $\delta$ increases, $D^\dis(\delta,L,B)$ turns to be positive earlier than $D^\dis(\delta,L,B)$.
            \item If $Q_H>Q_L$ and $\frac{p_Hc_H-p_Lc_L}{Q_H-Q_L}<\frac{p_Hc_H}{Q_B-Q_L}$, then the next action is \Htype, and finally B, which is exactly LH structure.
        \end{itemize}
        \item[(b)] Case 2: $\cel>\ceh$. Since the second action is \Htype, we can consider the following two functions.
        \begin{equation*}
            D^\dis(\delta,H,L) = (Q_L-Q_H)\left(F(\delta)-(1-\beta)\frac{p_Hc_H-p_Lc_L}{Q_H-Q_L}\right)
        \end{equation*}
        \begin{equation*}
            D^\dis(\delta,H,B)
            = (Q_B-Q_H)\left(F(\delta)-(1-\beta)\frac{p_Lc_L}{Q_B-Q_H}\right)
        \end{equation*}
        Hence, there are also three possibilities:
        \begin{itemize}
            \item If $Q_H>Q_L$, then the next action is \Both, which is exactly None-L structure. This is because there does not exist a $\delta$ satisfying that $D^\dis(\delta,H,L)>0$ and $D^\dis(\delta,N,H)>0$ when $Q_H>Q_L$ and $\cel>\ceh$. If otherwise, we have 
            \begin{eqnarray*}
                D^\dis(\delta,H,L)>0 \Rightarrow F(\delta) < (1-\beta)\frac{p_Hc_H-p_Lc_L}{Q_H-Q_L}
            \end{eqnarray*}
            \begin{eqnarray*}
                D^\dis(\delta,N,H)>0 \Rightarrow F(\delta) > (1-\beta) \ceh
            \end{eqnarray*}
            According to Eq.~\eqref{eq: eta2}, they are contradictory.
            \item If $Q_H<Q_L$ and $\frac{p_Hc_H-p_Lc_L}{Q_H-Q_L}>\frac{p_Lc_L}{Q_B-Q_H}$, the next action is \Both, which is exactly None-L structure.
            \item If $Q_H<Q_L$ and $\frac{p_Hc_H-p_Lc_L}{Q_H-Q_L}<\frac{p_Lc_L}{Q_B-Q_H}$, then the next action is \Ltype, and finally \Both, which is exactly HL structure.
        \end{itemize}
    \end{enumerate}
    
\end{proof}
        
\subsection{Proof for Lemma \ref{lemma: structure (1)}}

\begin{IEEEproof}
    Since the structure of $\Lambda$-$optimal$ policy is LH structure, we can find upper bounds of $\delta_L,\delta_H,\delta_B$ from $D^\dis(\delta, N, L)$, $D^\dis(\delta, L, H)$, and $D^\dis(\delta, H, B)$, respectively. 
    
    Recall that 
    \begin{equation*}
        D^\dis(\delta,N,L) = Q_L\left(\dis h_\dis(\delta+1) + \beta \ageloss -(1-\beta)\cel\right)
    \end{equation*}
    where $h_\dis(\delta+1)=J^\dis(\delta+1)-J^\dis(1)>0$. Thereby, $D^\dis(\delta,N,L)$ is definitely positive as $\delta>\sqrt{(1-\beta)/(\beta\cel)}$. Moreover, the threshold $\delta_L$ should be an integer. Hence, the company will never choose action \None, when AoI $\delta$ is larger than $\hat{\delta}_L=\left[\sqrt{(1-\beta)/(\beta\cel)}\right]^+$. And based on the recruitment order in LH structure, the next action should be \Ltype, therefore, the threshold $\delta_L$ of initiating the \Tl~vehicle recruitment should be less than $\hat{\delta}_L$, \emph{i.e.,} $\delta_L<\hat{\delta}_L$.
    
    Similarly, we can find the upper bounds of $\delta_H, \delta_B$ from $D^\dis(\delta,L,H)$, $D^\dis(\delta,H,B)$, respectively. 
\end{IEEEproof}

Here, we present the threshold bounds for HL, None-L, and None-H structures.
\begin{lemma}\label{lemma: structure (2)}
    Consider the HL structure of the optimal policy in Theorem~\ref{theorem: 1}, the optimal thresholds $\delta_H$, $\delta_L$, $\delta_B$ are upper bounded by the following values, respectively. 
    \begin{equation}
        \hat{\delta}_H=\left[\sqrt{\frac{(1-\beta)\ceh}{\beta}}\right]^+
    \end{equation}
    \begin{equation}
        \hat{\delta}_L=\left[\sqrt{\frac{(1-\beta)(\ceh Q_H-\cel Q_L)}{\beta(Q_H-Q_L)}}\right]^+
    \end{equation}
    \begin{equation}
        \hat{\delta}_B=\left[\sqrt{\frac{(1-\beta)\ceh}{\beta(1-Q_L)}}\right]^+
    \end{equation}
\end{lemma}
\begin{lemma}\label{lemma: structure (3)}
    Consider the None-L structure of the optimal policy in Theorem~\ref{theorem: 1}, the optimal thresholds $\delta_H$, $\delta_B$ are upper bounded by the following values, respectively.
    \begin{equation}
        \hat{\delta}_H=\left[\sqrt{\frac{(1-\beta)\ceh}{\beta}}\right]^+
    \end{equation}
    \begin{equation}
        \hat{\delta}_B=\left[\sqrt{\frac{(1-\beta)\cel}{\beta(1-Q_H)}}\right]^+
    \end{equation}
\end{lemma}
\begin{lemma}\label{lemma: structure (4)}
    Consider the None-H structure of the optimal policy in Theorem~\ref{theorem: 1}, the optimal thresholds $\delta_H$, $\delta_B$ are upper bounded by the following values, respectively. 
    \begin{equation}
        \hat{\delta}_L=\left[\sqrt{\frac{(1-\beta)\cel}{\beta}}\right]^+
    \end{equation}
    \begin{equation}
        \hat{\delta}_B=\left[\sqrt{\frac{(1-\beta)\ceh}{\beta(1-Q_L)}}\right]^+
    \end{equation}
\end{lemma}
The proofs for Lemma~\ref{lemma: structure (2)}-\ref{lemma: structure (4)} are broadly similar.

\subsection{Proof for Theorem~\ref{thm: alg converges}} \label{appendix: thm 2 and 3}

The proof roadmap is first to show that the minimum average cost of approximate MDP $\Lambda^{(m)}$ converges to that of original MDP $\Lambda$, and then to show that for a given number of truncated states $m$, the output of Algorithm~\ref{alg: BRVIA} converges to the optimal policy of MDP $\Lambda^{(m)}$. Based on these two, we can prove Theorem~\ref{thm: alg converges}.

Before we show the proof, let us give some useful lemmas.
\begin{lemma} \label{lemma: cost comparison}
    Under the $\Lambda_\dis$-$optimal$ policy $\phi^{\dis,*}$ and for any $t>0$ and $m>\hat{\delta}_B$, we have
    \begin{equation*}
        u(\delta^{(m)}(t),\phi^*(\delta^{(m)}(t))) \leq u(\delta(t),\phi^*(\delta(t)))
    \end{equation*}
\end{lemma}
\begin{proof}
    See Appendix F.
\end{proof}

\begin{theorem}
    Let $J^{*}, J^{(m),*}$ be the minimum average cost for the MDPs $\Lambda$ and $\Lambda^{(m)}$, respectively. Then, $J^{(m),*}\to J^{*}$ as $m\to\infty$.
\end{theorem}
\begin{IEEEproof}
Let $J^{\dis,(m)}(s)$ and $h_{\dis}^{(m)}(s)$ be the minimum expected total $\dis$-discount cost and the relative cost function, respectively. According to \cite{sennott1997computing}, we have to show the following two conditions are satisfied.
\begin{enumerate}
    \item There exists a nonnegative $L$, and nonnegative finite function $G(.)$ on $\mathcal{S}$ such that $-L \leq h_\dis^{(m)}(\delta) \leq G(\delta)$ for all $\delta\in\mathcal{S}^{(m)}$, where $m>\hat{\delta}_B$ and $0<\dis<1$: Again, $L=1$ is already safe since $h_\dis^{(m)}(\delta)$ is nonnegative. Similarly, we can say $J^{\dis,(m)}(\delta)>0$, thereby
    \begin{equation}\label{eq: eq1}
        h_\dis^{(m)}(\delta) = J^{\dis,(m)}(\delta) - J^{\dis,(m)}(1) < J^{\dis,(m)}(\delta).
    \end{equation}
    Next, suppose that $\phi^{\dis,*}$ is the optimal policy of MDP $\Lambda^{(m)}$, we have following inequality since $J^{\dis,(m)}(s)$ is minimum expected total $\dis$-discount cost and $\phi^{\dis,*}$ may not be the optimal policy for MDP $\Lambda^{(m)}_\dis$.
    \begin{equation}\label{eq: eq2}
        J^{\dis,(m)}(\delta) \leq J^{\dis,(m)}(\delta, \phi^{\dis,*}).
    \end{equation}
    Then, we will show $J^{\dis,(m)}(\delta, \phi^{\dis,*}) \leq J^{\dis,*}(\delta)$ for all $\delta\in\mathcal{S}^{(m)}$. Based on the definition, we have
    \begin{equation*}\label{eq: truncated J function}
        J^{\dis,(m)}(\delta, \phi^{\dis,*}) = \Ex \sum_{t=1}^\infty \beta^{t-1} u(\delta^{(m)}(t),\phi^{\dis,*}(\delta^{(m)}(t)))
    \end{equation*}
    \begin{equation*}\label{eq: J function}
        J^{\dis,*}(\delta) = \Ex \sum_{t=1}^\infty \beta^{t-1} u(\delta(t),\phi^{\dis,*}(\delta(t)))
    \end{equation*}
    
    Based on Lemma~\ref{lemma: cost comparison}, we can conclude for any $\delta\in\mathcal{S}^{m}$,
    \begin{equation}\label{eq: eq3}
        J^{\dis,(m)}(\delta, \phi^{\dis,*}) \leq J^{\dis,*}(\delta)
    \end{equation}
    Thereby, we can combine Eq.~\eqref{eq: eq1}\eqref{eq: eq2}\eqref{eq: eq3} to derive the following inequality.
    \begin{equation*}
        h_\dis^{(m)}(\delta) < J^{\dis,(m)}(\delta) \leq J^{\dis,(m)}(\delta, \phi^{\dis,*}) \leq J^{\dis,*}(\delta).
    \end{equation*}

    Based on the proof of Lemma~\ref{lemma: bounded value function} and the optimality of $\phi^{\dis,*}$, we have
    \begin{equation*}
        J^{\dis,*}(\delta) < \dis \left(\frac{-\delta}{\ln{\dis}} + \frac{1}{\ln^2{\dis}}\right).
    \end{equation*}
    Then we can let $G(\delta)=\dis \left(\frac{-\delta}{\ln{\dis}} + \frac{1}{\ln^2{\dis}}\right)$.
    \item The minimum average cost $J^{(m)}$ of MDP $\Lambda^{(m)}$ is bounded by $J^*$: Since we know $J^{\dis,(m)}(\delta) \leq J^{\dis,*}(\delta)$, and
    \begin{equation*}
        \begin{split}
            J^{(m),*} &= \lim\sup_{\dis\to1} (1-\dis) J^{\dis,(m)}(\delta)\\
            &\leq \lim\sup_{\dis\to1} (1-\dis) J^{\dis,*}(\delta) = J^*
        \end{split}        
    \end{equation*}
\end{enumerate}
\end{IEEEproof}

Next, we should show the following theorem.
\begin{theorem}
    For MDP $\Lambda^{(m)}$ with a given $m$, the output of Algorithm~\ref{alg: BRVIA} converges to the optimal policy of MDP $\Lambda^{(m)}$.
\end{theorem}
\begin{IEEEproof}
    According to \cite{puterman2014markov}, we only need to verify that the approximate MDP $\Lambda^{(m)}$ is unichain, i.e., the Markov chain corresponding to every deterministic stationary policy consists of a single recurrent class plus a possibly empty set of transient states. Since state $m$ is reachable from all other states, there is only one recurrent class by \cite{hsu2019scheduling}.
    
\end{IEEEproof}


\section*{Appendix A: Proof of Lemma~\ref{lemma: bounded value function}}
Consider the policy that always recruits nothing at each time slot. Then, the discount-weighted sum of AoI loss and recruitment under the aforementioned policy is the upper bound of the average cost under the optimal policy based on Eq.~\eqref{eq. discount MDP average cost}.
\begin{equation*}
    \begin{split}
        J^\dis(\delta) &\leq \beta(\delta^2+\beta(\delta+1)^2+\beta^2(\delta+2)^2+...)\\
            &\leq \beta \int_{x=0}^\infty (\delta+x)^2\dis^x = \beta \left(\frac{-\delta}{\ln{\dis}} + \frac{1}{\ln^2{\dis}}\right) < \infty,
    \end{split}    
\end{equation*}
which proves this lemma.

\section*{Appendix B: Proof of Lemma~\ref{lemma: the monotonicity value function}}

Utilizing the iteration algorithm, we can prove this lemma. At the beginning of the iteration, we initialize $J_0^\dis(\delta)=0$. Therefore, when $t=0$, the statement holds. Suppose that the statement still holds for $t=k$,
\begin{equation*}
    J_k^{\dis}(\delta) \leq J_k^{\dis}(\delta+1) \text{ for } \delta \geq 1
\end{equation*}
For $t=k+1$, we have
\begin{equation*}
    J_{k+1}^{\beta}(\delta) = \min_{a\in\mathcal{A}} u(\delta,a) + \beta \sum_{\delta'\in\mathcal{S}} P_{\delta\delta'}(a) J_k^{\dis}(\delta')
\end{equation*}
\begin{equation*}
    J_{k+1}^{\beta}(\delta+1) = \min_{a\in\mathcal{A}} u(\delta+1,a) + \beta \sum_{\delta'\in\mathcal{S}} P_{(\delta+1)\delta'}(a) J_k^{\beta}(\delta')
\end{equation*}
On the one hand, for any $a\in\mathcal{A}$, we have $u(\delta,a) \leq u(\delta+1,a)$. On the other hand, 
\begin{itemize}
    \item $a=N$, we have $\sum_{\delta'\in\mathcal{S}} P_{\delta\delta'}(a) J_k^{\beta}(\delta') = J_k^{\beta}(\delta+1) \leq  J_k^{\beta}(\delta+2)=\sum_{\delta'\in\mathcal{S}} P_{(\delta+1)\delta'}(a) J_k^{\beta}(\delta')$.
    \item $a=L$, we have $\sum_{\delta'\in\mathcal{S}} P_{\delta\delta'}(L) J_k^{\beta}(\delta') = Q_LJ_k^{\beta}(1) + (1-Q_L)J_k^{\beta}(\delta+1) \leq  Q_LJ_k^{\beta}(1) + (1-Q_L)J_k^{\beta}(\delta+2)=\sum_{\delta'\in\mathcal{S}} P_{(\delta+1)\delta'}(L) J_k^{\beta}(\delta')$.
    \item $a=H$, we have $\sum_{\delta'\in\mathcal{S}} P_{\delta\delta'}(H) J_k^{\beta}(\delta') = Q_HJ_k^{\beta}(1) + (1-Q_H)J_k^{\beta}(\delta+1) \leq  Q_HJ_k^{\beta}(1) + (1-Q_H)J_k^{\beta}(\delta+2)=\sum_{\delta'\in\mathcal{S}} P_{(\delta+1)\delta'}(H) J_k^{\beta}(\delta')$.
    \item $a=B$, we have $\sum_{\delta'\in\mathcal{S}} P_{\delta\delta'}(B) J_k^{\beta}(\delta') = Q_BJ_k^{\beta}(1) + (1-Q_B)J_k^{\beta}(\delta+1) \leq  Q_BJ_k^{\beta}(1) + (1-Q_B)J_k^{\beta}(\delta+2)=\sum_{\delta'\in\mathcal{S}} P_{(\delta+1)\delta'}(B) J_k^{\beta}(\delta')$.
\end{itemize}
Therefore, we can conclude that $J_{k+1}^{\beta}(\delta) \leq J_{k+1}^{\beta}(\delta+1)$. Moreover, $J_{t}^{\beta}(\delta)$ will converge to $J^{\beta}(\delta)$ as $t\to\infty$.

As a result, $J^{\beta}(\delta)$ is non-decreasing in $\delta$.

\section*{Appendix C: Proof of Lemma~\ref{lemma: from discount to average}}

According to \cite{hsu2019scheduling}, we need to show that the following two conditions are satisfied.
\begin{enumerate}
    \item There exists a deterministic stationary policy $\phi$ of the MDP $\Lambda$ such that the resulting discrete-time Markov chain (DTMC) by the policy is irreducible, aperiodic, and the average cost $J(\phi)$ is finite: Let $\phi$ be the deterministic stationary policy 'zero-wait', i.e., the company chooses to recruit vehicles of both types all the time. It is obvious that the resulting DTMC is irreducible and aperiodic. Next, we can compute the average cost under the zero-wait policy.
    \begin{equation*}
        \begin{split}
            J(\phi) &= \sum_{\delta=1}^\infty \pi(\delta)[(1-\beta)p_Bc_B + \beta(1-Q_B)\delta^2]\\
            &=\pi(1)\sum_{\delta=1}^\infty(1-Q_B)^{\delta-1}\left[(1-\beta)p_Bc_B + \beta(1-Q_B)\delta^2\right],
        \end{split}
    \end{equation*}
    where $\pi(\delta)$ is stationary distribution of the MDP $\Lambda$ under this policy, and $\pi(1) = Q_B$. The average cost $J(\phi)$ is finite, since it is the summation of a sequence that is the product of polynomials and exponentials and exponentials is based on $(1-Q_B)$ less than 1.
    \item There exists a nonnegative L such that the relative cost
    function $h_{\dis}(\delta) \geq -L$ for all $\delta$ and $\dis$: Since $J^{\dis}(\delta)$ is non-decreasing in $\delta$, we have $h_{\dis}(\delta) = J^{\dis}(\delta)-J^{\dis}(1) \geq 0$. We can just choose $L=1$.
\end{enumerate}
Let $\{\dis_n\}_{n=1}^\infty$ be a sequence of the discount factors converging to 1. According to \cite{sennott1989average}, if the above two conditions hold, there exists a subsequence $\{\theta_n\}_{n=1}^\infty$ such that a $\Lambda$-$optimal$ policy is the limit point of the $\Lambda_\beta$-$optimal$ policies. Therefore, the $\Lambda$-$optimal$ policy has the same structure and upper bounds of thresholds as $\Lambda_\beta$-$optimal$ policies when $\beta\to1$.

\section*{Appendix D: Proof of Lemma \ref{lemma: action N appears first and action B appears finally}}
If $\phi^{\beta,*}(1)\not=N$, then
\begin{equation*}
    Q^\beta(1,N) > \min\{Q^\beta(1,L),Q^\beta(1,H),Q^\beta(1,B)\}.
\end{equation*}
And $Q^\beta(2,N) - Q^\beta(1,N)<\min_{a\in\mathcal{A}/N}\{Q^\beta(2,a) - Q^\beta(1,a)\}$
since 
\begin{equation*}
    Q^\beta(2,N) - Q^\beta(1,N) = \beta + \beta(J^\beta(3)-J^\beta(2)) 
\end{equation*}
\begin{equation*}
    Q^\beta(2,a) - Q^\beta(1,a) = (1-Q_a)\beta + \beta(1-Q_a)(J^\beta(3)-J^\beta(2)) 
\end{equation*}
for $a=L,H,B$.
Hence, we have
\begin{equation*}
    Q^\beta(2,N) > \min\{Q^\beta(2,L),Q^\beta(2,H),Q^\beta(2,B)\}.
\end{equation*}
So on and so forth, we can conclude that, for all $\delta$,
\begin{equation*}
    Q^\beta(\delta,N) > \min\{Q^\beta(\delta,L),Q^\beta(\delta,H),Q^\beta(\delta,B)\}.
\end{equation*}

\section*{Appendix E: Proof of Lemma \ref{lemma: cost efficiency inequality}}

Recall that $\eta_L=\frac{c_L}{r_L}$, $\eta_H=\frac{c_H}{r_H}$, and $\eta_B=\frac{p_Lc_L+p_Hc_H}{p_Lr_L+p_Hr_H-Q_LQ_H}$. 
\begin{itemize}
    \item For Eq.~\eqref{eq: expected recruitment expenses per HD map}, If this inequality does not hold, then we have
    \begin{equation*}
        \eta_B<\min\{\eta_L,\eta_H\} \Rightarrow \eta_B < \eta_L \text{ and } \eta_B<\eta_H
    \end{equation*}
    We can find that $\eta_B < \eta_L$ implies that $\eta_L<\eta_H$ while $\eta_B < \eta_H$ implies that $\eta_L>\eta_H$, which is a contradiction. Hence, the above inequality does not hold, and the Eq.~\eqref{eq: expected recruitment expenses per HD map} holds.
    \item For Eq.~\eqref{eq: eta1}, since $Q_L>Q_H$ and $C_L<C_H$, we have 
    \begin{equation*}
        \begin{split}
            &\frac{p_Hc_H-p_Lc_L}{Q_H-Q_L}-C_L\\
           =&\frac{p_Hc_H-p_Lc_L-(Q_H-Q_L)C_L}{Q_H-Q_L}\\
           =&\frac{(p_Hc_H-p_Lc_L)r_L-(p_Hr_H-p_Lr_L)c_L}{(Q_H-Q_L)r_L}\\
           =&\frac{(p_Hc_H)r_L-(p_Hr_H)c_L}{(Q_H-Q_L)r_L}=\frac{p_Hr_H(C_H-C_L)}{Q_H-Q_L}<0,\\
        \end{split}        
    \end{equation*}
    which proves this inequality.
    \item For Eq.~\eqref{eq: eta2}, since $Q_L<Q_H$ and $C_L>C_H$, we have 
    \begin{equation*}
        \begin{split}
            &\frac{p_Hc_H-p_Lc_L}{Q_H-Q_L}-C_H\\
           =&\frac{p_Hc_H-p_Lc_L-(Q_H-Q_L)C_H}{Q_H-Q_L}\\
           =&\frac{(p_Hc_H-p_Lc_L)r_H-(p_Hr_H-p_Lr_L)c_H}{(Q_H-Q_L)r_H}\\
           =&\frac{(p_Lr_L)c_H-(p_Lc_L)r_H}{(Q_H-Q_L)r_H}=\frac{p_Lr_L(C_H-C_L)}{Q_H-Q_L}<0,\\
        \end{split}        
    \end{equation*}
    which proves this inequality.
\end{itemize}

\section*{Appendix F: Proof of Lemma \ref{lemma: cost comparison}}

    Notice that the vehicles' arrival and the quality of sensing data are independent of the state and policy. Let $\mathcal{q}=\{q(1),q(2),...,q(\infty)\}$ be the sequence of whether the company can collect the qualified sensing data if she recruits some vehicles at time slot $t$, i.e., $q(t)=1$ if the company can, otherwise, $q(t)=0$. For any $\mathcal{q}$, we can consider two mutually exclusive and complementary cases:
    \begin{itemize}
        \item Suppose that a state sequence is $\{\delta(1),\delta(2),...,\delta(\infty)\}$ where there does not exist any $\delta(i)=m$ for $i=1,2,...$. In that case, $\delta^{(m)}(t)=\delta(t)$, and thus, Eq.~\eqref{eq: truncated J function} is equal to Eq.~\eqref{eq: J function} for any $\delta\in\mathcal{S}^{m}$.
        \item Suppose that there exists some $i$ such that $\delta^{(m)}(i)=m$. We also need to show that for any $t>0$, we have
        \begin{equation*}
            u(\delta^{(m)}(t),\phi(\delta)) \leq u(\delta(t),\phi(\delta))
        \end{equation*}
        To prove the above inequality, we will show that $\delta^{(m)}(t)=[\delta(t)]^+_m$ for any $t$. Consider the first time that $\delta^{(m)}(i)=m$, then $\delta^{(m)}(i+k)=m$ and $\delta(i+k)=m+k$ for $k=0,1,...,K$ where $K$ is the smallest $k\geq0$ such that $q(i+k)=1$. In this sense, $\delta^{(m)}(t)=[\delta(t)]^+_m$ for $t=1,2,...,i,i+1,...,i+K$. When $t=i+K+1$, we can find $\delta^{(m)}(i+K+1)=1$ and $\delta(i+K+1)=1$ since $\phi(\delta)=B$ for all $\delta>m>\hat{\delta}_B$ and $q(i+k)=1$. Hence, $\delta^{(m)}(t)=[\delta(t)]^+_m$ for $t=i+K+1$. So on and so forth, we can conclude that $\delta^{(m)}(t)=[\delta(t)]^+_m$ for any $t>0$. Then, we can find that if $\delta(t)\leq m$, 
        \begin{equation*}
            u(\delta^{(m)}(t),\phi(\delta)) = u(\delta(t),\phi(\delta))
        \end{equation*}
        if $\delta(t)> m$,  
        \begin{equation*}
            u(\delta^{(m)}(t),\phi(\delta)) < u(\delta(t),\phi(\delta))
        \end{equation*}
        Consequently, for any $t>0$, we have
        \begin{equation*}
            u(\delta^{(m)}(t),\phi(\delta)) \leq u(\delta(t),\phi(\delta))
        \end{equation*}
    \end{itemize}

\end{document}